\documentclass[12pt]{article}
\usepackage{graphicx}%
\usepackage{multirow}%
\usepackage{amsmath,amssymb,amsfonts}%
\usepackage[backend=biber,bibstyle=numeric, sorting=ynt,maxbibnames=99]{biblatex}
\AtBeginBibliography{\footnotesize}

\usepackage{amsthm}%
\usepackage{hyperref}
\usepackage{mathrsfs}%
\usepackage[title]{appendix}%
\usepackage{xcolor}%
\usepackage{textcomp}%
\usepackage{manyfoot}%
\usepackage{booktabs}%
\usepackage{algorithm}%
\usepackage{algorithmicx}%
\usepackage{algpseudocode}%
\usepackage{listings}%
\usepackage{csquotes}
\usepackage{bbm}
\usepackage{notations}
\usepackage{authblk}
\usepackage{geometry}
\newcommand{\E}{\mathbb{E}}
\newcommand{\R}{\mathbb{R}}

\theoremstyle{plain} 
\newtheorem{theorem}{Theorem}

\newtheorem{lemma}[theorem]{Lemma}

\theoremstyle{definition}

\theoremstyle{remark} 
\newtheorem{remark}{Remark}

\title{Central limit theorem for the overlaps\\on the Nishimori line}

\author{Francesco Camilli$^{\dagger}$, Pierluigi Contucci$^{\diamond}$, Emanuele Mingione$^\diamond$}
\affil{\small $\dagger$ \emph{Abdus Salam International Center for Theoretical Physics, Italy}\\ 
$\diamond$ \emph{Dipartimento di Matematica, Alma Mater Studiorum - Università di Bologna}}

\begin{document}

\maketitle

{
\let\thefootnote\relax\footnotetext{\url{fcamilli@ictp.it}, \url{pierluigi.contucci@unibo.it}, \url{emanuele.mingione2@unibo.it}}
}

\vspace{-10pt}

\abstract{The overlap distribution of the Sherrington-Kirkpatrick model on the Nishimori line has been proved to be self averaging for large volumes. Here we study the joint distribution of the rescaled overlaps around their common mean and prove that it converges to a Gaussian vector.}

\section{Introduction}
Mean-field models of statistical mechanics are expected to be described by an order parameter and its fluctuation properties. In the deterministic cases, like the Curie-Weiss model and its variants, the magnetization concentrates almost everywhere at large volumes in the plane of the temperature and magnetic field i.e. it satisfies the law of large numbers (LLN). The critical half-line, low temperatures and zero external field, is the exception to such concentration and presents two coexisting thermodynamic phases with a spin-flip spontaneously broken symmetry. The presence of the two phases can be seen as a mild violation of the LLN. Outside the critical line, where the magnetization distribution converges to a delta function the fluctuations follow the central limit theorem (CLT), while at the critical point the distribution scales with the power 3/4 of the volume and converges to a quartic distribution \cite{EllisNewman78,ellis_book_largedev}.
The mean-field spin-glass case, the Sherrington Kirkpatrick (SK) model, described by the Parisi theory \cite{MPV}, has the peculiarity of a structural violation of the LLN at low temperature and small external field \cite{Pastur91}. The overlap, the order parameter of the model, is the scalar product between two independent spin configurations sampled from the Boltzmann-Gibbs measure with a fixed realization of the disorder. The permutation invariant overlap array among an arbitrary number of identical copies (real replicas) of the system has a non trivial limiting joint distribution characterized by specific factorization properties \cite{ACid,GG_original,GG_contucci_Giardina}, in particular ultrametricity \cite{MPV,baffioni2000some}. 
Parisi theory has been rigorously completed by a series of results \cite{Guerra_upper_bound,Tala_parisiformula,Panchenkoultra, 2015sherrinpanchenkogton}. At high temperatures and large magnetic fields, when the overlap concentrates \cite{Tala_vol2,ChenATsharp,Jagannathbasco}, it has been proved that the fluctuations are normal \cite{GuerraToninelliCLT, Tala_vol2}. 
We mention that, in  the presence of quenched disorder, also  the  fluctuations  of the free energy play an important role in the characterization of the thermodynamic properties of the system \cite{ALR88,Chattebook, PanWeDey}. 
The region where the overlap concentrates extends also beyond the one we mentioned. An important role is played by the region where one allows random interactions with positive mean: 
\begin{align}
    \label{Ham_SK}
    H_N(\boldsymbol{\sigma})= - \sum_{i<j=1}^N \tilde{z}_{ij}\sigma_i\sigma_j
\end{align}
where $\tilde{z}_{ij}\iid\mathcal{N}(\lambda/N,\tilde{\lambda}/N)$, for some $\lambda,\tilde{\lambda}>0$ and $\boldsymbol{\sigma}=(\sigma_i)_{i\leq N}\in\{-1,+1\}^N$, i.e. the sum of Curie-Weiss and SK models. The free energy of such model, that turns out to have a Parisi-like representation, has been studied in \cite{Chenferromagnetic14} and generalized in \cite{Camilli_mismatch}. 

In this paper we focus on the celebrated special case defined by $\lambda=\tilde{\lambda}$ which coincides with the \textit{Nishimori line}, that emerged for the first time in the gauge theory of spin-glasses (see \cite{nishimori01} for a classical reference). Within such line the model fulfills a set of remarkable properties consisting of identities and correlation inequalities \cite{contucci_giardina_2012,Nishi_id_PC,contucci_morita_nishimori}. Among the consequences of those properties, there is the control of the asymptotic behaviour of the overlap, that turns out to be self-averaging \cite{MSKNL,DBMNL,adaptive,Barbier_2019,overlap_jean}, as predicted by the replica theory.
The relevance of model \eqref{Ham_SK} is also related to the remarkable correspondence between genuine statistical mechanics models and high-dimensional inference problems \cite{nishimori01,mezard2009information}. For example, the model \eqref{Ham_SK} on the Nishimori line corresponds to an instance of the spiked Wigner model \cite{Johnstone_WSM} defined as the problem of the reconstruction of a binary signal $\boldsymbol{\sigma}^*\in\{+1,-1\}^N$, from the set of noisy observations
\begin{align}\label{channel00}
    y_{ij}=\sqrt{\frac{\lambda}{N}}\sigma_i^*\sigma_j^*+z_{ij}\qquad i<j, \,,\quad z_{ij}\iid\mathcal{N}(0,1)
\end{align}with some $\lambda >0$.
The Bayesian approach suggests to study the posterior measure for $\boldsymbol{\sigma}^*$ given the $y_{ij}$'s. Exploiting the Gaussian nature of the noise $(z_{ij})_{\leq i<j\leq N}$, we can write the posterior as
\begin{align}
    P(\boldsymbol{\sigma}\mid \mathbf{y})=\frac{1}{\mathcal{Z}(\mathbf{y})}\exp\Big[-\frac{1}{2}\sum_{1\leq i<j\leq N}\Big(y_{ij}-\sqrt{\frac{\lambda}{N}}\sigma_i\sigma_j\Big)^2\Big]\,.
\end{align}
Re-absorbing trivial terms into the normalization, that is a partition function, and plugging \eqref{channel00} into the previous equation, one readily gets
\begin{align}
    P(\boldsymbol{\sigma}\mid \mathbf{y})=\frac{1}{\mathcal{Z}(\mathbf{z},\boldsymbol{\sigma}^*)}\exp\Big[\sum_{1\leq i<j\leq N}\Big(\sqrt{\frac{\lambda}{N}}z_{ij}+\frac{\lambda}{N}\sigma_i^*\sigma_j^*\Big)\sigma_i\sigma_j\Big]\,.
\end{align}
The latter can also be interpreted as a random Boltzmann-Gibbs measure, and as such we know that many of its features are encoded in its \emph{pressure}
\begin{align}
    p_N(\lambda)=\frac{1}{N}\mathbb{E}_{\mathbf{z},\boldsymbol{\sigma}^*}\log\mathcal{Z}(\boldsymbol{\sigma}^*,\mathbf{z})\,.
\end{align}
Notice now that after the $\mathbb{Z}_2$ gauge transformation $ 
    z_{ij}\mapsto z_{ij}\sigma_i^*\sigma_j^*\,,\quad \sigma_i\mapsto\sigma_i\sigma_i^*\,,
$ one actually gets rid of the original signal $\boldsymbol{\sigma}^*$ in the pressure:
\begin{align}
    p_N(\lambda)=\frac{1}{N}\mathbb{E}_{\mathbf{z}}\log\sum_{\boldsymbol{\sigma}\in\{-1,1\}^N}\exp\Big[\sum_{1\leq i<j\leq N}\Big(\sqrt{\frac{\lambda}{N}}z_{ij}+\frac{\lambda}{N}\Big)\sigma_i\sigma_j\Big]\,,
\end{align}
The random variable inside the round parenthesis is exactly a Gaussian with mean equal to its variance, which characterizes the Nishimori line in spin-glasses. 
The analogy is  not limited to the example above, but it can be extended to other types of spins, and many techniques coming from rigorous Statistical Mechanics, such as the cavity method \cite{Tala_vol1,ASS}, or the interpolation scheme \cite{interp_guerra_2002,Guerra_upper_bound} can actually be successfully transferred to inference \cite{COJAOGHLANcavity,adaptive}.

In the present work we study the rescaled and centered overlap array distribution on the Nishimori line for spins with bounded support measure and prove that it converges to a Gaussian vector for large volumes. Previous work on the topic had computed its second moment and proved a CLT for the free energy \cite{ElAlaoui}. Here we compute the entire overlap joint distribution by controlling the convergence of the limiting characteristic function.

The paper is organised as follows. Section \ref{sec:Setting} introduces the notations and the main result. Section \ref{sec:tools} presents the cavity method and the preliminary concentration results needed to obtain a convergence in distribution. Section \ref{sec:proof2repplica} proves the main theorem, first for the single overlap, then generalized in Section \ref{sec:proofmain} to an arbitrary number of replicas. The Appendix collects some technical parts of the proof of the main theorem.

\section{Definitions and main result}\label{sec:Setting}
From the Statistical Mechanics point of view, the model we want to study is completely characterized by the following Hamiltonian
\begin{align}
    \label{eq:Hamiltonian}
    -H_N(\mathbf{x};\mathbf{x}^*,\mathbf{z})=\sum_{i<j,1}^N\Big[\sqrt{\frac{\lambda}{N}}z_{ij}x_ix_j+\frac{\lambda}{N}x_ix_i^*x_jx_j^*-\frac{\lambda}{2N}x_i^2x_j^2\Big]\,,
\end{align}
where $z_{ij}=z_{ji}\iid\mathcal{N}(0,1)$ and $x_i^*\iid P_X$. The $\mathbf{z}$ and $\mathbf{x}^*$ are random variables that play the role of quenched disorder, so for a given realization of them the state of the system is determined by the $\mathbf{x}$. We further take $P_X$ with bounded support, though this assumption could be relaxed, and we endow the site variables $x_i$, or the spins, with the same apriori measure $P_X$.

Once the Hamiltonian is given, according to Boltzmann's prescription, we have a corresponding measure on the site variables, the Boltzmann-Gibbs measure, defined as follows:
\begin{align}
    \label{eq:BGmeasure}
    d\mu_N(\mathbf{x};\mathbf{z},\mathbf{x}^*):=\frac{dP_X(\mathbf{x})}{\mathcal{Z}_N(\mathbf{z},\mathbf{x}^*)}\exp\Big(-H_N(\mathbf{x};\mathbf{x}^*,\mathbf{z})\Big)\,,\quad dP_X(\mathbf{x})\equiv\prod_{i=1}^NdP_X(x_i)\,.
\end{align}
The normalization, denoted with $\mathcal{Z}$, will be referred to as partition function, and it depends on the quenched randomness $(\mathbf{z},\mathbf{x}^*)$. Its logarithm, re-scaled by the number of sites, defines an object called pressure per particle, or simply pressure, in Statistical Mechanics literature:
\begin{align}
    p_N(\mathbf{z},\mathbf{x}^*)=\frac{1}{N}\log\mathcal{Z}(\mathbf{z},\mathbf{x}^*)\,,
\end{align}
not to be confused with the pressure of a gas. Notice that both expectations w.r.t. the measure \eqref{eq:BGmeasure} and the pressure are random objects. One could indeed take a further expectation w.r.t. the quenched disorder. If we use the standard bracket notation for Boltzmann averages, \emph{i.e.} for any, say bounded, test function $f(\mathbf{x})$ of the site variables
\begin{align}
    \langle f\rangle_{\mathbf{z},\mathbf{x}^*}=\int d\mu_N(\mathbf{x};\mathbf{z},\mathbf{x}^*)f(\mathbf{x})
\end{align}
the so-called quenched averages would appear as
\begin{align}
    \mathbb{E}_{\mathbf{z},\mathbf{x}^*}\langle f \rangle_{\mathbf{z},\mathbf{x}^*}=\mathbb{E}_{\mathbf{z},\mathbf{x}^*}\int d\mu_N(\mathbf{x};\mathbf{z},\mathbf{x}^*)f(\mathbf{x})\,.
\end{align}
Taking the expectation on the random pressure per particle gives instead rise to what shall be referred to as quenched pressure per particle, or free entropy
\begin{align}
    \Bar{p}_N(\lambda)=\frac{1}{N}\mathbb{E}_{\mathbf{z},\mathbf{x}^*}\log\mathcal{Z}(\mathbf{z},\mathbf{x}^*)\,.
\end{align}
To state our main result we also need to define the replicated Boltzmann-Gibbs measure, that will be needed to average over $n\in\mathbb{N}$ independent \emph{i.i.d.} samples from the Gibbs measure:
\begin{align}
    \label{eq:replicatedBGmeasure}
    d\mu_N^{\bigotimes n}((\mathbf{x}^{(a)})_{a=1}^n;\mathbf{z},\mathbf{x}^*)=\prod_{a=1}^n\frac{dP_X(\mathbf{x}^{(a)})}{\mathcal{Z}(\mathbf{z},\mathbf{x}^*)}\exp\Big(-H_N(\mathbf{x}^{(a)};\mathbf{z},\mathbf{x}^*)\Big)\,.
\end{align}
With an abuse of notation we shall denote expectations w.r.t. the replicated measure still with $\langle\cdot\rangle_{\mathbf{z},\mathbf{x}^*}$. Notice that replicas, namely the i.i.d. samples, share the same quenched disorder. Hence they are independent only conditionally on $\mathbf{z},\mathbf{x}^*$, and a further quenched expectation w.r.t. the latter would couple them. In other words, for any bounded test functions $f(\mathbf{x}^{(1)},\dots,\mathbf{x}^{(j)})$, $g(\mathbf{x}^{(j+1)},\dots,\mathbf{x}^{(n)})$ one has
\begin{multline}
    \mathbb{E}_{\mathbf{z},\mathbf{x}^*}\langle f(\mathbf{x}^{(1)},\dots,\mathbf{x}^{(j)})g(\mathbf{x}^{(j+1)},\dots, \mathbf{x}^{(n)})\rangle_{\mathbf{z},\mathbf{x}^*}=\mathbb{E}_{\mathbf{z},\mathbf{x}^*}\langle f(\mathbf{x}^{(1)},\dots,\mathbf{x}^{(j)})\rangle_{\mathbf{z},\mathbf{x}^*}\langle g(\mathbf{x}^{(j+1)},\dots, \mathbf{x}^{(n)})\rangle_{\mathbf{z},\mathbf{x}^*}\\
    \neq \mathbb{E}_{\mathbf{z},\mathbf{x}^*}\langle f(\mathbf{x}^{(1)},\dots,\mathbf{x}^{(j)})\rangle_{\mathbf{z},\mathbf{x}^*}\mathbb{E}_{\mathbf{z},\mathbf{x}^*}\langle g(\mathbf{x}^{(j+1)},\dots, \mathbf{x}^{(n)})\rangle_{\mathbf{z},\mathbf{x}^*}\,.
\end{multline}
In order to lighten the notation, from now on the subscripts $\mathbf{z},\mathbf{x}^*$ will be dropped.
Notice that the distribution of the overlap function is  invariant  under permutations even at finite volumes.
The main quantity under investigation is the \emph{overlap} family
\begin{align}
    \label{eq:overlap_family}
    q_{ab}=\frac{1}{N}\sum_{i=1}^Nx_{i}^{(a)}x_i^{(b)}\,,\quad 0\leq a <b\leq n
\end{align}
where the $0$-th replica is by convention $\mathbf{x}^*$. 

The model defined by the Hamiltonian \eqref{eq:Hamiltonian} can be regarded as an extension of the SK model on the Nishimori line to other type of spins. The Hamiltonian \eqref{eq:Hamiltonian} indeed arises in inference in the so-called spiked Wigner model. In this very popular inference problem, a Statistician has to retrieve a rank-one spike, of the form $\mathbf{x}^*\mathbf{x}^{*\intercal}/\sqrt{N}$, from the set of observations
\begin{align}\label{eq:WSchannel}
    y_{ij}=\sqrt{\frac{\lambda}{N}}x_i^*x_j^*+z_{ij}\,,\quad 1\leq i<j\leq N\,.
\end{align}
If the Statistician is \emph{optimal}, meaning they know everything about the generating process of the $y_{ij}$ ($P_X$, $\lambda$, the relation \eqref{eq:WSchannel}), then they have access to the Bayes posterior
\begin{align}
    dP_{\mathbf{x^*|\mathbf{y}}}(\mathbf{x})\propto dP_X(\mathbf{x})\exp\Big[-\frac{1}{2}\sum_{i<j}^N\Big(\sqrt{\frac{\lambda}{N}}x_ix_j-y_{ij}\Big)^2\Big]\,.
\end{align}
By plugging \eqref{eq:WSchannel} into the previous equation, after computing the square one realizes that the posterior measure is nothing but the Boltzmann-Gibbs measure \eqref{eq:BGmeasure}, with a normalization constant being exactly $\mathcal{Z}(\mathbf{z},\mathbf{x}^*)$. The optimality of the Statistician carries on some properties that are inherited by expectations of observables called Nishimori identities, that can be stated as follows. Given a bounded test function $f(\mathbf{y},\mathbf{x}^*,(\mathbf{x})_{a=2}^{n})$, it holds that:
\begin{align}
    \label{eq:NishimoriID}
    \mathbb{E}\langle f(\mathbf{y},\mathbf{x}^*,(\mathbf{x})_{a=2}^{n})\rangle=\mathbb{E}\langle f(\mathbf{y},(\mathbf{x})_{a=1}^{n})\rangle\,,
\end{align}
namely we can replace $\mathbf{x}^*$ with an additional independent replica. This has already some interesting consequences on the overlap family:
\begin{align}
    \mathbb{E}\langle g(q_{ab}) \rangle=\mathbb{E}\langle g(q_{10}) \rangle\,,\quad  1\leq a<b\leq n\,,
\end{align}
for any bounded test function $g$. The Nishimori identities, in all of their generality, can be seen as a consequence of Bayes rule. A quick proof can be found in \cite{Lelarge2017FundamentalLO}. In this setting with Gaussian disorder $\mathbf{z}$ it is possible to obtain them, starting from the r.h.s. of \eqref{eq:NishimoriID}, directly through the change of variables
\begin{align}
    z_{ij}\mapsto z_{ij}-\sqrt{\frac{\lambda}{N}}x_i^*x_i^*\,.
\end{align}

The thermodynamic limit of our model, and its generalizations were intensively studied in recent years \cite{lesieur2015mmse,el2018estimation,adaptive,ElAlaoui} with rigorous approaches, that produced a replica symmetric formula for the limit of $\bar{p}_N(\lambda)$. The term \enquote{replica symmetric} refers to the fact that the order parameter of the model concentrates around its expected value, which in turn implies a finite dimensional variational principle for the quenched pressure. For this model the order parameter is the overlap between $\mathbf{x}^*$ and $\mathbf{x}$, which converges to its expectation in the thermodynamic limit in a proper sense \cite{overlap_jean}. More specifically, the following result holds:
\begin{align}
    \label{eq:RSformula}
    p(\lambda):=\lim_{N\to\infty}\bar{p}_N(\lambda)=\sup_{q\geq 0}\Big\{
    -\frac{\lambda q^2}{4}+\mathbb{E}\log\int dP_X(x)e^{(\sqrt{\lambda q}z+\lambda q x^*)x-\frac{\lambda q x^2}{2}}
    \Big\}\,,
\end{align}
with $z\sim\mathcal{N}(0,1)$. Similarly to what happens for spin-glass models on the Nishimori line, the optimality of the Statistician constrains $\Bar{p}_N(\lambda)$ to be non-decreasing and convex in $\lambda$ \cite{contucci_morita_nishimori}, and so will be its limit. Hence $p(\lambda)$ is twice differentiable almost everywhere by Alexandroff's lemma. It can be proved that if $p$ is twice differentiable at a given $\lambda$, then the supremum is uniquely attained at a value $\bar{q}(\lambda)$ satisfying the consistency equation
\begin{align}
    \label{eq:consistency}
    q=F(\lambda q)\,,\quad F(r)=\mathbb{E}\,x^*\frac{\int dP(x)\,x\,e^{(\sqrt{r}z+r x^*)x-\frac{r}{2}x^2}}{\int dP(x)e^{(\sqrt{r}x+r x^*)x-\frac{r}{2}x^2}}\equiv\mathbb{E}\,x^*\langle x\rangle_{r,z,x^*}\,,
\end{align}
and $\bar{q}(\lambda)$ can be identified as the asymptotic expected overlap between $\mathbf{x}^*$ and $\mathbf{x}$ \cite{Lelarge2017FundamentalLO}.

The optimization w.r.t. $q$ can give rise to phase transitions in the model, that are typically controlled by $\lambda$. From the inferential point of view, $\lambda$ is the signal-to-noise ratio. With the present scalings, with a signal that has a diverging number of components to be estimated, it can indeed happen that if $\lambda$ is not big enough it is not possible to retrieve the spike signal. In the $N\to\infty$ limit this threshold will be given by
\begin{align}
    \label{eq:deflambda_C}
    \lambda_c=\sup\{\lambda >0 \,:\, \bar{q}(\lambda)>0\}\,.
\end{align}
Notice that for some kind of priors, for instance not centered ones, we could have $\lambda_c=0$. In fact, in order to obtain a non-trivial alignment with $\mathbf{x}^*$ it would be sufficient to sample directly from $P_X$.

The very same model could be given with some additional one-body terms that do not break the Nishimori identities \eqref{eq:NishimoriID}, as follows:
\begin{align}
    \label{eq:Hamiltonian_complete}
    \begin{split}
        -\tilde H_N(\mathbf{x};\mathbf{x}^*,\mathbf{z},\mathbf{h})&=\sum_{i<j,1}^N\Big[\sqrt{\frac{\lambda}{N}}z_{ij}x_ix_j+\frac{\lambda}{N}x_ix_i^*x_jx_j^*-\frac{\lambda}{2N}x_i^2x_j^2\Big]\\
        &+\sum_{i=1}^N\Big(\sqrt{h}h_ix_i+hx_i^*x_i-\frac{h x_i^2}{2}\Big)\,,
    \end{split}
\end{align}
where $h_i\iid\mathcal{N}(0,1)$ and independent on $\mathbf{z}$. It is easy to verify then that the associated Boltzmann-Gibbs measure is the posterior measure of a spiked Wigner model with the additional \enquote{side information}:
\begin{align}
    \tilde y_i=\sqrt{h}x_i^*+h_i\,.
\end{align}
For this model everything is unchanged, including uniqueness and convexity properties, except for the variational formula for the pressure,
\begin{align}
    p(\lambda,h):=\lim_{N\to\infty}\bar{p}_N(\lambda,h)=\sup_{q\geq 0}\Big\{
    -\frac{\lambda q^2}{4}+\mathbb{E}\log\int dP_X(x)e^{(\sqrt{\lambda q+h}z+(\lambda q+h) x^*)x-\frac{(\lambda q+h) x^2}{2}}
    \Big\}\,,
\end{align}
and the consistency equation, that becomes
\begin{align}
    q=F(\lambda q+h)\,,
\end{align}
whose solution, when unique, is denoted by $\bar{q}(\lambda,h)$.
This variant of the original problem \eqref{eq:Hamiltonian}  will be used only in the proofs.

We can finally state our main result, that is a Central Limit Theorem for the family of the overlaps:
\begin{theorem}\label{th:main}
Consider the model with Hamiltonian \eqref{eq:Hamiltonian}, $\lambda$ in the set where the convex function $p(\lambda)$ is twice differentiable, and the random array 
\begin{equation}\label{rarray}
\boldsymbol{\xi}=(\xi_{ab})_{1\leq a<b\leq n}\,,\quad \xi_{ab}=\sqrt{N}(q_{ab}-\bar{q}(\lambda)).
\end{equation}
If either $\lambda <\lambda_c$ or $P_X$ is not symmetric about the origin then 
\begin{equation}
 \boldsymbol{\xi}\xrightarrow[N\to\infty]{\mathcal{D}}\mathcal{N}(0,\Sigma)
\end{equation}
w.r.t. the quenched measure $\E\langle\cdot\rangle$, for a suitable covariance matrix
 $\Sigma=(\Sigma_{ab,cd})$ with $1\leq a<b\leq n, 1\leq c<d\leq n$.
\end{theorem}

\begin{remark}
The permutation invariance among replicas is not broken by the $N\to\infty$ limit. This constrains the elements of the covariance matrix $\Sigma$ to take only three possible values: $\Sigma_{ab,ab}=A$, $\Sigma_{ab,bd}=B$, $\Sigma_{ab,cd}=C$ with $a\neq c$ and $b\neq d$, that are computed in \eqref{covariances}.
\end{remark}

\section{Methods and preliminary results}\label{sec:tools}
In this section we illustrate the main tools that lead to the proof of Theorem \ref{th:main}. The whole proof revolves around the so called \emph{cavity method}, also known sometimes as \emph{leave-one-out method}.  We will  follow Talagrand's construction \cite{Tala_vol1} of the cavity interpolation introduced in the analysis  of the Sherringhton-Kirkpatrick model and subsequently adapted to the inference setting by \cite{ElAlaoui}. Here we repeat this construction for the help of the reader. 

The idea is to study the variation, along a suitable interpolating path, of the expected values of some observables when one spin of the system is \enquote{isolated}, namely decoupled from the other ones.
It turns out that the most convenient choice is an interpolating  path that leaves  the model on the Nishimori line  for any value of $t$.  Since for models like  \eqref{eq:Hamiltonian} the Nishimori line condition is not as explicit as for binary spins, for which it is sufficient to have Gaussian interactions with mean equal to their variance, we need to exploit the inference point of view and construct the Gibbs measure as a posterior. For $t\in[0,1]$ and  $q\geq0$, that will be fixed later on, let us consider the following set of observations: 
\begin{align}\label{channels}
    &\begin{cases}
    y_{ij}=\sqrt{\frac{\lambda }{N}}x_i^*x_j^*+z_{ij}\quad i<j\leq N-1\\
    y_{i,N}=\sqrt{\frac{\lambda t}{N}}x_i^*x_N^*+z_{iN}\quad i\leq N-1\\
    y_N=\sqrt{\lambda(1-t)q}x_N^*+z_N
    \end{cases}
\end{align}
where $z_N$ and $(z_{ij})_{i<j\leq N}$ are all independent copies of a standard Gaussian. 
Given the observations \eqref{channels}, the posterior measure on the $\mathbf{x}$'s is then
\begin{multline}
    dP_{\mathbf{X}|\mathbf{Y}=\mathbf{y}}(\mathbf{x})\propto \prod_{i<j\leq N-1}\exp\left\{
    -\frac{1}{2}\left(\sqrt{\frac{\lambda}{N}}x_ix_j-y_{ij}\right)^2
    \right\}\\
    \times\prod_{i\leq N}\exp\left\{
    -\frac{1}{2}\left(
    \sqrt{\frac{t\lambda}{N}}x_ix_N-y_{iN}
    \right)^2
    \right\}\times\exp\left\{-\frac{1}{2}\left(
    \sqrt{\lambda(1-t)q}x_N-y_N
    \right)^2\right\}\,.
\end{multline}
Expanding the squares and reabsorbing into the normalization those terms that do not depend on $\mathbf{x}$ we get
\begin{multline}
    dP_{\mathbf{X}|\mathbf{Y}=\mathbf{y}}(\mathbf{x})\propto
    \exp\left\{
    \sum_{i<j}^{N-1}\left[\sqrt{\frac{\lambda}{N}}y_{ij}x_ix_j-\frac{\lambda}{2N}x_i^2x_j^2\right]\right.\\\left.+\sum_{i=1}^N\left[
    \sqrt{\frac{t\lambda}{N}}y_{iN}x_ix_N-\frac{\lambda t}{2N}x_i^2x_N^2
    \right]+
    \sqrt{\lambda(1-t)q}y_N x_N-\frac{\lambda(1-t)q}{2}x_N^2
    \right\}\,.
\end{multline}
Now, if we plug \eqref{channels} into the previous equation we can finally identify the cavity Hamiltonian:
\begin{multline}\label{eq:cavity_Hamiltonian}
    -H_t(\mathbf{x})=\sum_{i<j}^{N-1}\left[\sqrt{\frac{\lambda}{N}}z_{ij}x_ix_j+\frac{\lambda}{N}x_i^*x_ix_j^*x_j-\frac{\lambda}{2N}x_i^2x_j^2\right]\\
    +\sum_{i=1}^N\left[
    \sqrt{\frac{t\lambda}{N}}z_{iN}x_ix_N+\frac{t\lambda}{N}x_ix_i^*x_N x_N^*-\frac{\lambda t}{2N}x_i^2x_N^2
    \right]\\+
    \sqrt{\lambda(1-t)q}z_N x_N+\lambda(1-t)qx_Nx_N^*-\frac{\lambda(1-t)q}{2}x_N^2\,.
\end{multline}
It is clear that for $t=1$ we recover the original model \eqref{eq:Hamiltonian} while for $t=0$  the last particle $x_N$ is decoupled from the rest of the system and replaced by a random Gaussian external field whose variance is proportional to $q$, plus a contribution in the direction of the ground truth component $x_N^*$ again proportional to $q$. Let us now tune the value of $q=\bar{q}(\lambda)$ where $\bar{q}(\lambda)$ is the solution of the variational problem \eqref{eq:RSformula}. Notice that, by hypothesis of Theorem \ref{th:main} if $p(\lambda)$ is twice differentiabile at $\lambda$, then  $\bar{q}(\lambda)$ is uniquely defined and satisfies the consistency equation \eqref{eq:consistency}.

The main point is that by construction, the interpolating model is on the Nishimori line for any $t\in[0,1]$, namely the quenched measure $\mathbb{E}\langle \cdot\rangle_t$ satisfies the Nishimori identities \eqref{eq:NishimoriID}, where $\langle\cdot \rangle_t$, with abuse of notation, denotes also the replicated random Boltzmann-Gibbs measure induced by the cavity Hamiltonian \eqref{eq:cavity_Hamiltonian}.

\subsection{Properties of the cavity measure}
In this section we study the properties of the quenched measure $\mathbb{E}\langle\cdot\rangle_t$ that will be denoted for brevity as $\nu_t(\cdot)$. In what follows, since it has a special role,  the last component $x_N^{(l)}$ of the vector $\mathbf{x}^{(l)}$ for  $l=*,1,\dots,n$, will be denoted by  $\epsilon^{(l)}$. We recall that we set $0\equiv\ast$ for replica indices. We also need to introduce a slight modification of the overlap $q_{\ell\ell'}$ in \eqref{eq:overlap_family} :

\begin{equation}\label{eq:defovermen}
 q_{\ell\ell'}^{-}:=\frac{1}{N}\sum_{i=1}^{N-1}x_i^{(\ell)}x_i^{(\ell')}= q_{\ell\ell'}-\frac{1}{N} \epsilon^{(\ell)}\epsilon^{(\ell')}
\end{equation}
for $\ell,\ell'\in\{*,1,\ldots\,,n\}$. The next lemma allows to control how the measure $\nu_t(\cdot)$ changes along the time parameter $t$.
\begin{lemma}\label{lemma_dv}
For any bounded function $f$ of $n$ replicas of the signal and the ground truth the following holds:
\begin{multline}\label{derivata_nu}
\frac{\mathrm{d}}{\mathrm{d} t} \nu_{t}(f)=\frac{\lambda}{2} \sum_{1 \leq l \neq l^{\prime} \leq n} \nu_{t}\left(\left(q_{l, l^{\prime}}^{-}-\bar{q}\right) \epsilon^{(l)} \epsilon^{\left(l^{\prime}\right)} f\right)-\lambda n \sum_{l=1}^{n} \nu_{t}\left((q_{l, n+1}^{-}-\bar{q}\right) \epsilon^{(l)} \epsilon^{(n+1)} f )\\
+\lambda \sum_{l=1}^{n} \nu_{t}\left(\left(q_{l, *}^{-}-\bar{q}\right) \epsilon^{(l)} \epsilon^{*} f\right)-\lambda n \nu_{t}\left(\left(q_{n+1, *}^{-}-\bar{q}\right) \epsilon^{(n+1)} \epsilon^{*} f\right) \\
+\lambda \frac{n(n+1)}{2} \nu_{t}\left(\left(q_{n+1, n+2}^{-}-\bar{q}\right) \epsilon^{(n+1)} \epsilon^{(n+2)} f\right)\,.
\end{multline}
\end{lemma}

\begin{proof}
The proof essentially follows from \cite[Lemma 1.4.2]{Tala_vol1} that is based on Gaussian integration by parts. However, our Hamiltonian contains two sources of disorder: the Gaussian noise in the $\mathbf{z}$'s, and the ground truth $\mathbf{x}^*$. Let us introduce the auxiliary Hamiltonian
\begin{multline}
    -H_{t,s}(\mathbf{x})=\sum_{i<j}^{N-1}\left[\sqrt{\frac{\lambda}{N}}z_{ij}x_ix_j+\frac{\lambda}{N}x_i^*x_ix_j^*x_j-\frac{\lambda}{2N}x_i^2x_j^2\right]\\
    +\sum_{i=1}^N\left[
    \sqrt{\frac{s\lambda}{N}}z_{iN}x_ix_N+\frac{t\lambda}{N}x_ix_i^*x_N x_N^*-\frac{\lambda t}{2N}x_i^2x_N^2
    \right]\\+
    \sqrt{\lambda(1-s)q}z_N x_N+\lambda(1-t)qx_Nx_N^*-\frac{\lambda(1-t)q}{2}x_N^2\,.
\end{multline}
For the $s$-derivative of the corresponding quenched measure $\nu_{t,s}$ we can use \cite[Lemma 1.4.2]{Tala_vol1}:
\begin{multline}
    \frac{d}{ds}\nu_{t,s}(f)=\frac{\lambda}{2}\sum_{1\leq l,l'\leq n}\nu_{t,s}\left(f\epsilon^{(l)}\epsilon^{(l')}\left(q_{ll'}^--\bar{q}\right)\right)-n\lambda\sum_{l=1}^n\nu_{t,s}\left(f\epsilon^{(l)}\epsilon^{(n+1)}\left(q_{l,n+1}^--\bar{q}\right)\right)\\
    -n\frac{\lambda}{2}\nu_{t,s}\left(
    f \epsilon^{(n+1)\,2}\left(q_{n+1,n+1}^--\bar{q}\right)
    \right)+\frac{\lambda n (n+1)}{2}\nu_{t,s}\left(f \epsilon^{(n+1)}\epsilon^{(n+2)}\left(q_{n+1,n+2}^--\bar{q}\right)\right)\,.
\end{multline}
The $t$-derivative instead requires no integration by parts and can be computed directly:
\begin{multline}
    \frac{d}{dt}\nu_{t,s}(f)=\lambda\sum_{l=1}^n\nu_{t,s}\left(f \epsilon^{(l)}\epsilon^*(q_{l*}^--\bar{q})\right)-\lambda n\nu_{t,s}\left(f \epsilon^{(n+1)} \epsilon^* (q_{n+1,*}^--\bar{q})\right)\\
    -\frac{\lambda}{2}\sum_{l=1}^n\nu_{t,s}\left(f\epsilon^{(l)2}(q_{ll}^--\bar{q})
    \right)+n\frac{\lambda}{2}\nu_{t,s}\left(
    f \epsilon^{(n+1)\,2}\left(q_{n+1,n+1}^--\bar{q}\right)
    \right)\,.
\end{multline}
Adding up the two previous contributions and setting $s=t$ we get the statement.
\end{proof}
From the previous lemma we can get a (still raw) control on the expectation of a non-negative function:

\begin{lemma}
For any bounded non-negative function $f$ of $n$ signal replicas and the ground truth we have
\begin{align}\label{raw_estimate}
    \nu_t(f)\leq K(\lambda,n)\nu_{1}(f)\,,
\end{align}
where $K(\lambda,n)>0$ is a constant depending only on the signal-to-noise ratio $\lambda$ and the number of replicas $n$.
\end{lemma}

\begin{proof}
Recall that we consider bounded support priors. Therefore using \eqref{derivata_nu} and the triangular inequality it is easy to see that $|\nu'_t(f)|\leq C(\lambda,n)\nu_t(f)$ for some constant $C(\lambda,n)$ depending only on the signal-to-noise ratio $\lambda$ and the number of replicas $n$, and then
\begin{align}
    \nu'_t(f)\geq -C(\lambda,n)\nu_t(f)\,.
\end{align}
The rest follows from an application of Gronwall's lemma.
\end{proof}

Now we come to the central technical result that allows us to carry out the entire cavity computation.
\begin{lemma}
For any bounded function $f$ of $n$ replicas and the ground truth we have
\begin{align}\label{main_estimate1}
    &\left|\nu_t(f)-\nu_0(f)\right|\leq K(\lambda,n)\nu^{1/\tau_1}\left(|q_{1\ast}^--\bar{q}|^{\tau_1}\right)\nu^{1/\tau_2}\left(|f|^{\tau_2}\right)\\
    \label{main_estimate2}
    &\left|\nu_t(f)-\nu_0(f)-\nu'_{0}(f)\right|\leq K(\lambda,n)\nu^{1/\tau_1}\left(|q_{1\ast}^--\bar{q}|^{2\tau_1}\right)\nu^{1/\tau_2}\left(|f|^{\tau_2}\right)\,,
\end{align}
for any non-negative $\tau_1,\tau_2$ such that $\frac{1}{\tau_1}+\frac{1}{\tau_2}=1$. 
\end{lemma}
\begin{remark}
We stress that in \eqref{main_estimate2} $\tau_1$ is multiplied by a factor $2$ at the exponent of $|q_{1\ast}^--\bar{q}|$, which proves to be crucial to obtain the Central Limit Theorem from concentration of the overlap, later displayed in \eqref{eq:concentration_elalaoui}, at least with the approach used in this work.
\end{remark}
\begin{proof}
Thanks to Lagrange's mean value theorem
\begin{align}
    \left|\nu_t(f)-\nu_0(f)\right|\leq\sup_{t\in[0,1]}|\nu'_t(f)|\,.
\end{align}
Using \eqref{derivata_nu}, the fact that $P_X$ has bounded support and the triangular inequality one can bound the first derivative of $\nu_t$ with an expression containing terms of the type
\begin{align}
    \nu_t\left(|f|\left|q_{ll'}^--\bar{q}\right|\right)
\end{align}
that by H\"older are further bounded by
\begin{align}
    \nu_t^{1/\tau_1}\left(|q_{ll'}^--\bar{q}|^{\tau_1}\right)\nu_t^{1/\tau_2}\left(|f|^{\tau_2}\right)\,,
\end{align}
with $\tau_1$ and $\tau_2$ as in the statement. After an application of the Nishimori identities and \eqref{raw_estimate} any term used to bound the derivative has the same form. Therefore inequality \eqref{main_estimate1} is proved.

Concerning \eqref{main_estimate2}, notice that
\begin{align}
    \left|\nu_t(f)-\nu_0(f)-\nu'_{0}(f)\right|\leq\sup_{t\in[0,1]}|\nu_t''(f)|\,.
\end{align}
So now the goal is to bound the second derivative of $\nu_t$. This can be done applying \eqref{derivata_nu} twice. Using again the triangular inequality one can bound it with an expression containing terms like
\begin{align}
    \nu_t\left(
    |f|\left|q_{ll'}^--\bar{q}\right|\left|q_{rr'}^--\bar{q}\right|
    \right)\leq\nu_t^{1/\tau_1}\left(\left|q_{ll'}^--\bar{q}\right|^{\tau_1}\left|q_{rr'}^--\bar{q}\right|^{\tau_1}\right)\nu_t^{1/\tau_2}\left(|f|^{\tau_2}\right)\,.
\end{align}
To obtain the final bound one can use Cauchy-Schwartz's inequality on the first expectation on the r.h.s. The result then follows again by the Nishimori identities and \eqref{raw_estimate}.
\end{proof}
\subsection{Overlap concentration}
We borrow a concentration result from \cite[Theorem 7]{ElAlaoui} for the 
overlap $q_{1*}$ with respect to the measure $\nu_t$ induced by the cavity Hamiltonian \eqref{eq:cavity_Hamiltonian}:
\begin{theorem}\label{thm:concalaoui}
Suppose $p$ is twice differential at a given $\lambda$. Then for all $t\in[0,1]$ there exist two constants $K(\lambda)\geq 0$ and $c(t)$ such that
\begin{align}
    \label{eq:concentration_elalaoui}
    \nu_t\Big((q_{1\ast}-\bar{q}(\lambda))^4\Big)\leq K(\lambda)\Big(\frac{1}{N^2}+e^{-c(t)N}\Big)\,,
\end{align}
with $c(t)>0$ on $[0,1)$. Moreover if either $\lambda<\lambda_c$ or $P_X$ is not symmetric about the origin,
then $c(t)>0$ on $[0,1]$.
\end{theorem}
We note that by Jensen's inequality the previous gives also an analogous control on the quadratic fluctuations of $q_{1\ast}$ around $\bar{q}(\lambda)$:
\begin{align}\label{eq:concentration_quadratic}
    \nu_t\Big((q_{1\ast}-\bar{q}(\lambda))^2\Big)\leq K(\lambda)\Big(\frac{1}{N}+e^{-c(t)N}\Big)\,.
\end{align}
for some $K$ and $c$ with the same features as above.\\

\subsection{Lemmata}

We have already seen that the cavity method leads to expected values of observable containing the  quantity $(q^-_{ll}-\bar{q}(\lambda))$ where $q^-_{ll}$ is defined in \eqref{eq:defovermen}. We also introduce the rescaled version of them:
\begin{equation}\label{rarraymeno}
 \xi_{ll'}^-=\sqrt{N}(q^-_{ll'}-\bar{q}(\lambda))\,,\quad l,l'\in\{*,1,\ldots\,n\}, \boldsymbol{\xi}^-=(\xi^-_{ab})_{1\leq a<b\leq n}
\end{equation}
The expectations simplified via the cavity method will present one missing spin ($t=0$ in the interpolation). Hence, in order to close some consistency equations, we will need to add it back and control the error, which is precisely the purpose of the following Lemma.
\begin{lemma}\label{replacement_lemma}
Given $\mathbf{u}\in\R^{n(n-1)/2}$   and  $l,l'\in\{*,1,\ldots\,n+2\}$ consider the  quantity
\begin{equation}
\nu\left(\xi_{ll'}
e^{i\mathbf{u}^\intercal\boldsymbol{\xi}}\right)
\end{equation}
where $\nu$ is the quenched measure over $n+2$ replicas. Then
\begin{equation}\label{eq:estcavi}
\big|\nu\left(\xi_{ll'}
e^{i\mathbf{u}^\intercal\boldsymbol{\xi}}\right)-\nu_0\left(\xi^-_{ll'}
e^{i\mathbf{u}^\intercal\boldsymbol{\xi}^-}\right) \big|\leq \dfrac{C(1+|\mathbf{u}|)}{\sqrt{N}} 
\end{equation}
where $C$ is independent from $N$ and $u$.
\end{lemma}

\begin{proof}
For $t\in[0,1]$ and $l,l'$ as above let us define
\begin{equation}\label{very}
\xi_{ll'}(t)=\xi_{ll'}^- +\frac{t}{\sqrt{N}}x_N^{(l)}x^{(l)}_N\,
\end{equation}
and $\boldsymbol{\xi}(t)=(\xi_{ab}(t))_{1\leq a<b'\leq n}$. Define also the function
\begin{equation}
g(t)=\nu_t\Big(\xi_{ll'}(t)e^{i\mathbf{u}^\intercal\boldsymbol{\xi}(t)} \Big)
\end{equation}
where $\nu_t$ is the quenched measure on $n+2$ replicas  induced by the cavity Hamiltonian \eqref{eq:cavity_Hamiltonian}.
Since $g(1)$  and $g(0)$ are respectively  equal to the two terms in the l.h.s. of \eqref{eq:estcavi},   if we show that $\Big|\dfrac{dg}{dt}\Big| \leq C\dfrac{1+|\mathbf{u}|}{\sqrt{N}} $ uniformly in $t$ then \eqref{eq:estcavi} is proved. A direct computation gives
\begin{equation}\label{eq:esti}
\frac{dg}{dt}=\frac{d\nu_t}{dt}\Big(\xi_{ll'}(t)e^{i\mathbf{u}^\intercal\boldsymbol{\xi}(t)}\Big)+\nu_t\Big(\frac{d\xi_{ll'}(t)}{dt}e^{i\mathbf{u}^\intercal\boldsymbol{\xi}(t)}\Big)+i\nu_t\Big(\xi_{ll'}(t)\mathbf{u}^{\intercal} \dfrac{d
\boldsymbol{\xi}(t)}{dt}e^{i\mathbf{u}^\intercal\boldsymbol{\xi}(t)}\Big)\,.
\end{equation}
Concerning the last term, by the very definition \eqref{very} of $\xi_{ab}(t)$, the fact that all the $x$'s are bounded random variables and by \eqref{eq:concentration_quadratic} we have that
\begin{align}
    \Big|\nu_t\Big(\xi_{ll'}(t)\mathbf{u}^{\intercal} \dfrac{d
\boldsymbol{\xi}(t)}{dt}e^{i\mathbf{u}^\intercal\boldsymbol{\xi}(t)}\Big)\Big|\leq \dfrac{C|\mathbf{u}|}{\sqrt{N}}\,.
\end{align}
The second term in the above equation is bounded in norm simply by $C/\sqrt{N}$, again by boundedness of the $x$'s and of the imaginary exponential. The first term  in \eqref{eq:esti} can be computed using  \eqref{derivata_nu} and using Cauchy-Schwartz inequality that leads to a bound by a sum of terms proportional to the quantities 
\begin{equation}
\Big[\nu_t\left((\xi_{ll'}(t))^2\right) \nu_t\left((q_{ll'}^--\bar{q})^2\right)\Big]^{1/2} .
\end{equation}
Let us recall that by \eqref{eq:concentration_quadratic} 
$\nu_1\left((\xi_{l,l'})^2\right)< \infty$,
therefore using \eqref{raw_estimate} one also has  $ \nu_t\left((\xi_{l,l'})^
2\right)< \infty$
and then $ \nu_t\left((q_{ll'}^--\bar{q})^2\right)=O\left(N^{-1}\right)$. This implies that also  the first term  in \eqref{eq:esti} is $O(N^{-1/2})$ and this concludes the proof.
\end{proof}

In the proof of Theorem \ref{th:main} we will show that the  characteristic function of  the random vector $\boldsymbol{\xi}$  defined in \eqref{rarray}  satisfies the same  differential equation of the characteristic function of a Gaussian vector up to a remainder that vanishes as $N$ goes to infinity. The next lemma gives sufficient conditions in order to conclude that if so, then the large $N$ limit of the characteristic function of $\boldsymbol{\xi}$ coincides with the characteristic function  of a Gaussian vector.

\begin{lemma}\label{lem:limit_ODE}
Let $(f_N)$ be a sequence of absolutely continuous, complex valued functions $f_N:\mathbb{R}\to \mathbb{C} $  that satisfy
\begin{equation}\label{ode}
\begin{cases}
f'_N(s)= - a\, s \,f_N(s)+r_N(s) ,\\
f_N(0)=1
\end{cases}
\end{equation}
where $a\in \R$ is a fixed parameter and  $r_N$ is  such that  $|r_N(s)|\leq g(s)N^{-\alpha}$ for some $\alpha>0$ with $g$  continuous, then 
\begin{equation}
\lim_{N\to\infty} f_N(s)=e^{-\frac{1}{2} a  s^2 }\,.
\end{equation}
\end{lemma}

\begin{proof}
Let us write $f_N=Re(f_N)+i Im(f_N)=:x_N+iy_N$, then the Cauchy problem  \eqref{ode} is equivalent to the system
\begin{equation}\label{odet2}
\begin{cases}
 x'_N(s)=- a\, s\,x_N(s) + Re(r_N(s))\\
x_N(0)=1
\end{cases}
\begin{cases}
y'_N(s)=- a\, s \,y_N(s) + Im(r_N(s))\\
y_N(0)=0
\end{cases}\,.
\end{equation}

We are going to prove that  $x_N(s)\to e^{-\frac{a s^2}{2}}$. The same argument can be used to show $y_N(s)\to 0$ concluding the proof of the lemma. Consider the function
$\Delta_N(s)=x_N(s)-e^{-\frac{a s^2}{2}}$. Then one can verify that
\begin{equation}\label{odet3}
\begin{cases}
 \Delta'_N(s)=- a\, s\,\Delta_N (s) + Re(r_N(s))\\
\Delta_N(0)=0
\end{cases}\,.
\end{equation}
Therefore
\begin{equation}
\frac{d}{ds}\big|\Delta_N(s)\big|\leq\big|\frac{d}{ds}\Delta_N(s)\big|\leq |as| \,\big|\Delta_N(s)\big|+  N^{-\alpha} g(s)  
\end{equation}
and then by Gronwall's inequality (see the classical reference \cite{evans10} for instance) one gets
\begin{equation}
\big|\Delta_N(s)\big|\leq N^{-\alpha}e^{|a|\frac{s^2}{2}}\int_0^s g(t) dt
\end{equation}
that implies 
\begin{equation}
\lim_{N\to\infty} \big|\Delta_N(s)\big|=0   
\end{equation}
that is
\begin{equation}
\lim_{N\to\infty} x_N(s)=e^{-\frac{a s^2}{2}}\,.
\end{equation}

\end{proof}

As we anticipated in Section \ref{sec:Setting}, we will also need to treat the model with an external magnetic field. There is a way to interpolate in the $(\lambda,h)$ plane from the original model to one with only external fields, \emph{i.e.} one-body terms in the Hamiltonian and keeping the expected overlap constant along the trajectory. The following Theorem, that can be considered an extension of \cite[Theorem 1]{quadratic_replica_coupling} characterizes these trajectories.

\begin{theorem}\label{overlap_invarianceThm}
    Let $\lambda$ be in the full Lebesgue measure where $p(\lambda)$ is twice differentiable. Recall the definition of $F$ in \eqref{eq:consistency}. Moreover, let $h(\lambda')=h_0-\lambda' F(h_0)$ for any $\lambda'\in[0,\lambda]$ be a trajectory in the plane $(\lambda,h)\in\mathbb{R}_{>0}\times \mathbb{R}_{\geq0}$. The following statements hold:
    \begin{enumerate}
        \item[1)] $F$ is non-decreasing in $\mathbb{R}\geq0$;
        \item[2)] the supremum of $p(\lambda)$ is uniquely attained at a point $\bar{q}(\lambda)\equiv q(\lambda,0)$ s.t.
        \begin{align}
            q(\lambda,0)=F(\lambda q(\lambda,0))\,;
        \end{align}
        \item[3)] there exists a unique $h_0$ s.t. $h(\lambda)=0$ and $F(h_0)=q(\lambda,0)$, namely $h_0$ is uniquely identified by $\lambda$;
        \item[4)] for any $\lambda'\in[0,\lambda]$ the solution to the fixed point equation
        \begin{align}\label{FPintermediate}
            q(\lambda',h(\lambda'))=F(\lambda' q(\lambda',h(\lambda'))+h(\lambda'))
        \end{align}
        equals $F(h_0)$.
    \end{enumerate}
\end{theorem}

\begin{proof}
\phantom{.}

\emph{1).} This result is a correlation inequality analogous to the one obtained in \cite{contucci_morita_nishimori}. Using the Nishimori identities it is in fact possible to prove that:
\begin{align}
    \frac{dF}{dr}(r)=\mathbb{E}\Big(\langle X^2\rangle-\langle X\rangle^2\Big)^2\geq 0\,,
\end{align}
where we dropped subscripts for brevity.

\emph{2)}. See statement and proof of \cite[Theorem 1]{Lelarge2017FundamentalLO}.

\emph{3)}. $h(\lambda)=0$ means 
\begin{align*}
    h_0=\lambda F(h_0)\,.
\end{align*}
It is easy to realize that $h_0=\lambda q(\lambda,0)$ both solves the previous equation and satisfies $F(h_0)=q(\lambda,0)$. Furthermore, it is uniquely identified once $\lambda$ is chosen by point \emph{2)}. Before going to the next point we stress that, without the condition $F(h_0)=q(\lambda,0)$, $h_0$ could also take other values, like $h_0=0$ in case of a symmetric prior. However, this would not yield the overlap $q(\lambda,0)$ which is instead needed to attain the supremum of the pressure if we are above the critical threshold $\lambda_c$.

\emph{4)}. Consider a point of the straight line $h(\lambda')=h_0-\lambda' F(h_0)$ for a given $\lambda'$. Then we have another equation for $h_0$: 
\begin{align*}
    h_0=h(\lambda')+\lambda' F(h_0)
\end{align*}
which is solved by
\begin{align}
    h_0=\lambda' q(\lambda',h(\lambda'))+h(\lambda')\,.
\end{align}
Notice that we have no guarantee up to this point that the quantity $q(\lambda',h(\lambda'))$ is single valued, since its related fixed point equation may have multiple solutions. Nevertheless \emph{any} of them must yield the same $h_0$, that is fixed by the end point $(\lambda,0)$ for any $\lambda'$. Hence, if we evaluate $F$ on both sides of the previous equation, using \eqref{FPintermediate} we easily get
\begin{align}
    q(\lambda,0)=F(h_0)=F\big(\lambda' q(\lambda',h(\lambda'))+h(\lambda')\big)=q(\lambda',h(\lambda'))
\end{align}
which finally implies that the solution is constant along the trajectory and it is uniquely determined by the end point $(\lambda,0)$.
\end{proof}

\section{Proof of Theorem \ref{th:main}, the case $n=2$}\label{sec:proof2repplica}

In this section we will prove Theorem \ref{th:main} in the case $n=2$, namely a CLT theorem for the rescaled overlap $\xi_{12}$ defined in \eqref{rarray}. In this simpler case the proof does not contain many  technicalities and  shows the main ideas used also in the general case. The proof relies on Lévy's continuity theorem. Hence, the main object under study will be the characteristic function:
\begin{align}
    \phi_N(u)=\nu\Big(e^{iu\xi_{12}}\Big)\,.
\end{align}
Now, thanks to the Nishimori identities \eqref{eq:NishimoriID} we can replace one of the two replicas with the quenched variable $\mathbf{x}^*$, this simplifies computations. We will then use the definition
\begin{align}\label{def_characeristic}
    \phi_N(u)=\nu\Big(e^{iu\xi_{1\ast}}\Big)\,.
\end{align}
In order to compute the $N\to\infty$ asymptotics of it  we use the cavity method to obtain a differential equation for \eqref{def_characeristic}. 
The $u$-derivative of $\phi_N$ yields:
\begin{multline}\label{partial_phi_cavity}
    \partial_u\phi_N(u)=i\nu \left(\xi_{1\ast}e^{iu\xi_{1\ast}}\right)=i\sqrt{N}\nu \left((\epsilon\epsilon^\ast- \bar{q} )e^{iu\xi_{1\ast}^-+iu\frac{\epsilon\epsilon^\ast}{\sqrt{N}}}\right)=\\=
    i\sqrt{N}\nu \left((\epsilon\epsilon^\ast- \bar{q})e^{iu\xi_{1\ast}^-}
    \left(1+iu\frac{\epsilon\epsilon^\ast}{\sqrt{N}}\right)\right)+O\left(\frac{1}{\sqrt{N}}\right)=\\=i\sqrt{N}\nu \left((\epsilon\epsilon^\ast- \bar{q} )e^{iu\xi_{1\ast}^-}\right)-u\,\nu \left(\epsilon\epsilon^\ast(\epsilon\epsilon^\ast- \bar{q} )e^{iu\xi_{1\ast}^-}\right)+O\left(\frac{1}{\sqrt{N}}\right)\,,
\end{multline}
where $\xi_{1\ast}^-$ is defined in \eqref{rarraymeno}. We stress that in the second equality we have used the permutation invariance among the signal components, and  $O\left(\frac{1}{\sqrt{N}}\right)$ has to be intended as a quantity that can be uniformly bounded in norm by $ C_u/\sqrt{N}$ for some $C_u>0$ at fixed $u$. In the present case the dependency of $C_u$ on $u$ is at most quadratic.

Let us treat the the second term first. By \eqref{main_estimate1} with $\tau_1=1,\tau_2=\infty$ one has
\begin{align}
    \nu\left(\epsilon\epsilon^\ast(\epsilon\epsilon^\ast-\bar{q})e^{iu\xi_{1\ast}^-}\right)=\nu_0\left(\epsilon\epsilon^\ast(\epsilon\epsilon^\ast-\bar{q})\right)\nu_0\left(e^{iu\xi_{1\ast}^-}\right)+\delta
\end{align}
where $\delta=O(\nu(|q_{1\ast}^--\bar{q}|))$, that thanks to \eqref{eq:concentration_elalaoui} is $\delta=O(1/\sqrt{N})$. Recall indeed that we are in the hypothesis of asymmetric prior or $\lambda<\lambda_c$ so $c(1)>0$ and the exponentially decaying contribution is negligible w.r.t. $1/\sqrt{N}$. This remainder can be considered independent of $u$ because the imaginary exponential is always bounded by $1$ in norm.

The first term in the last line of \eqref{partial_phi_cavity} requires more attention because of the $\sqrt{N}$ factor in front of it. Therefore, we will need the finer estimate \eqref{main_estimate2} with $\tau_1=1,\tau_2=\infty$ to obtain negligible remainders:
\begin{align}\label{ssstep}
    \begin{split}
       \nu \left((\epsilon\epsilon^\ast- \bar{q} )e^{iu\xi_{1\ast}^-}\right)=\nu_0 \left((\epsilon\epsilon^\ast- \bar{q} )e^{iu\xi_{1\ast}^-}\right) +\nu_0' \left((\epsilon\epsilon^\ast- \bar{q} )e^{iu\xi_{1\ast}^-}\right)+\delta'
    \end{split}
\end{align}
where $\delta'=O(\nu((q_{1\ast}^--\bar{q})^2))=O(1/N)$ uniformly in $u$ by \eqref{eq:concentration_quadratic} for the same reasons as above. Since at $t=0$ the spin  $\epsilon$ is decoupled from anything else, the first expectation on the r.h.s. of \eqref{ssstep} can be shown to vanish, in fact
\begin{equation}
    \nu_0\left(\epsilon\epsilon^*\right)\equiv\mathbb{E}\,x^*\langle x\rangle_{\lambda\bar{q},z,x^*}
\end{equation}
where $\langle x\rangle_{r,z,x^*}$
is defined in \eqref{eq:consistency} and then  $\nu_0\left(\epsilon\epsilon^*\right)=\bar{q}$.  To evaluate $\nu'_0$ we need instead \eqref{derivata_nu} with $f=(\epsilon\epsilon^\ast- \bar{q} )e^{iu\xi_{1\ast}^-}$, and in particular $n=1$:
\begin{multline}
    \nu_0' \left((\epsilon\epsilon^\ast- \bar{q} )e^{iu\xi_{1\ast}^-}\right)=-\lambda \nu_0\left(\epsilon^{(1)}\epsilon^{(2)}(q_{12}^--\bar{q})
    (\epsilon^{(1)}\epsilon^\ast- \bar{q} )e^{iu\xi_{1\ast}^-}
    \right)\\+\lambda \nu_0\left(\epsilon^{(1)}\epsilon^{\ast}(q_{1\ast}^--\bar{q})(\epsilon^{(1)}\epsilon^\ast- \bar{q} )e^{iu\xi_{1\ast}^-}\right)-\lambda \nu_0\left(
    \epsilon^{(2)}\epsilon^{\ast}(q_{2\ast}^--\bar{q})
    (\epsilon^{(1)}\epsilon^\ast- \bar{q} )e^{iu\xi_{1\ast}^-}
    \right)\\
    +\lambda \nu_0\left(
    \epsilon^{(2)}\epsilon^{(3)}(q_{23}^--\bar{q})
    (\epsilon^{(1)}\epsilon^\ast- \bar{q} )e^{iu\xi_{1\ast}^-}
    \right)\,.
\end{multline}
We recall that at $t=0$ the spin $\epsilon$ and its replicas are decoupled from the rest.  This allows us to factorize them and to average over them separately, thus obtaining the following coefficients: 
\begin{align}
\begin{split}
    &a_1:=\mathbb{E}\langle x^2\rangle^2_{\lambda\bar{q},z,x^*}-\bar{q}^2\\    
    &a_2:=\mathbb{E}\langle x^2\rangle_{\lambda\bar{q},z,x^*}\langle x\rangle^2_{\lambda\bar{q},z,x^*}-\bar{q}^2\\
    &a_3:=\mathbb{E}\langle x\rangle_{\lambda\bar{q},z,x^*}^4-\bar{q}^2\,.
\end{split}
\end{align}
With these notations we get
\begin{multline}
    \nu_0' \left((\epsilon\epsilon^\ast-\bar{q} )e^{iu\xi_{1\ast}^-}\right)=-2\lambda a_2 \nu_0\left((q_{12}^--\bar{q})e^{iu\xi_{1\ast}^-}
    \right)+\lambda a_1\nu_0\left((q_{1\ast}^--\bar{q})e^{iu\xi_{1\ast}^-}\right)\\+\lambda 
    a_3\nu_0\left((q_{23}^--\bar{q})
    e^{iu\xi_{1\ast}^-}
    \right)=\frac{\lambda a_1}{\sqrt{N}}\nu_0\left(\xi_{1\ast}^-e^{iu\xi_{1\ast}^-}\right)-2\frac{\lambda a_2}{\sqrt{N}}\nu_0\left(\xi_{12}^-e^{iu\xi_{1\ast}^-}\right)+\frac{\lambda a_3}{\sqrt{N}}\nu_0\left(\xi_{23}^-e^{iu\xi_{1\ast}^-}\right)\,.
\end{multline}
Now we plug everything into \eqref{partial_phi_cavity}, including the remainders, obtaining:
\begin{multline}\label{partial_u_phi}
    \partial_u\phi_N(u)=-u a_1\nu_0\left(e^{iu\xi_{1\ast}^-}\right)+ i\lambda a_1\nu_0\left(\xi_{1\ast}^-e^{iu\xi_{1\ast}^-}\right)-2i \lambda a_2 \nu_0\left(\xi_{12}^-e^{iu\xi_{1\ast}^-}\right)\\+i \lambda a_3 \nu_0\left(\xi_{23}^-e^{iu\xi_{1\ast}^-}\right)+O\left(\frac{1}{\sqrt{N}}\right)\,.
\end{multline}

We see that, even if we assume to be able to replace $\nu_0\left(e^{iu\xi_{1\ast}^-}\right)$ with $\phi_N(u)$, the equation is not self-contained, because of the presence of the last two terms. This suggests that we have to introduce two other functions to be treated with the cavity approach:
\begin{align}
    \label{auxiliary_field1}
    &\psi_N(u)=i\nu\left(\xi_{2*}e^{iu\xi_{1*}}\right)\\
    \label{auxiliary_field2}
    &\zeta_N(u)=i\nu\left(\xi_{23}e^{iu\xi_{1*}}\right)\,.
\end{align}
Let us further define the quantities
\begin{align}
    \label{phi_0-}
    &\phi_0^-(u)=\nu_0\left(e^{iu\xi_{1*}^-}\right)\\
    \label{psi_0-}
    &\psi_0^-(u)=i\nu_0\left(\xi_{2*}^-e^{iu\xi_{1*}^-}\right)\\
    \label{zeta_0-}
    &\zeta_0^-(u)=i\nu_0\left(\xi_{23}^-e^{iu\xi_{1*}^-}\right)\,.
\end{align}
With these notations \eqref{partial_u_phi} rewrites as 
\begin{align}
\begin{split}
    \partial_u\phi_N(u)=-u a_1\phi_0^-(u)+ \lambda a_1\partial_u \phi_0^-(u)-2 \lambda a_2 \psi_0^-(u)+ \lambda a_3 \zeta_0^-(u)+O\left(\frac{1}{\sqrt{N}}\right)\,.
\end{split}
\end{align}

Let us now turn to $\psi_N(u)$. Again by permutation symmetry of the spins
\begin{multline}\label{psi_cavity}
    \psi_N(u)=i\sqrt{N}\nu \left((\epsilon^{(2)}\epsilon^\ast- \bar{q} )e^{iu\xi_{1\ast}^-+iu\frac{\epsilon^{(1)}\epsilon^\ast}{\sqrt{N}}}\right)=\\=
    i\sqrt{N}\nu \left((\epsilon^{(2)}\epsilon^\ast-\bar{q} )e^{iu\xi_{1\ast}^-}
    \left(1+iu\frac{\epsilon^{(1)}\epsilon^\ast}{\sqrt{N}}\right)\right)+O\left(\frac{1}{\sqrt{N}}\right)=\\=i\sqrt{N}\nu \left((\epsilon^{(2)}\epsilon^\ast-\bar{q} ) e^{iu\xi_{1\ast}^-}\right)-u\,\nu \left(\epsilon^{(1)}\epsilon^\ast
    (\epsilon^{(2)}\epsilon^\ast-\bar{q} ) e^{iu\xi_{1\ast}^-}\right)+O\left(\frac{1}{\sqrt{N}}\right)\,.
\end{multline}
Now we use \eqref{main_estimate1} for the second term:
\begin{multline}
    \nu \left(\epsilon^{(1)}\epsilon^\ast
    (\epsilon^{(2)}\epsilon^\ast-\bar{q} ) e^{iu\xi_{1\ast}^-}\right)=\nu_0\left(\epsilon^{(1)}\epsilon^\ast
    (\epsilon^{(2)}\epsilon^\ast-\bar{q} )\right)\nu_0\left(e^{iu\xi_{1\ast}^-}\right)+\gamma
    =a_2\phi_0^-(u)+\gamma\,.
\end{multline}
Here $\gamma=O(\nu(|q_{1*}^--\bar{q}|))=O(1/\sqrt{N})$. The more delicate term that multiplies $\sqrt{N}$ has to be treated with \eqref{main_estimate2} where $f=(\epsilon^{(2)}\epsilon^\ast-\bar{q} ) e^{iu\xi_{1\ast}^-}$ which is a function of $n=2$ replicas and the ground truth:
\begin{align}
    \nu\left((\epsilon^{(2)}\epsilon^\ast-\bar{q} ) e^{iu\xi_{1\ast}^-}\right)=\nu_0(\epsilon^{(2)}\epsilon^\ast-\bar{q} ) \nu_0\left(e^{iu\xi_{1\ast}^-}\right)+\nu_0'\left((\epsilon^{(2)}\epsilon^\ast-\bar{q} ) e^{iu\xi_{1\ast}^-}\right)+\gamma'\,.
\end{align}
The first term on the r.h.s. vanishes since $\nu_0(\epsilon^*\epsilon^{(2)})=\bar{q}$. $\nu_0'$ can be evaluated through \eqref{derivata_nu} as before with $n=2$, yielding:
\begin{multline}
    \nu_0'\left((\epsilon^{(2)}\epsilon^\ast-\bar{q} ) e^{iu\xi_{1\ast}^-}\right)=\lambda a_2\nu_0\left((q_{12}^--\bar{q}) e^{iu\xi_{1\ast}^-}\right)-2\lambda a_3\nu_0\left((q_{13}^--\bar{q})e^{iu\xi_{1\ast}^-}\right)\\
    -2\lambda a_2\nu_0\left((q_{23}^--\bar{q})e^{iu\xi_{1\ast}^-}\right)+\lambda a_2 \nu_0\left((q_{1*}^--\bar{q})e^{iu\xi_{1\ast}^-}\right)+\lambda a_1\nu_0\left((q_{2*}^--\bar{q})e^{iu\xi_{1\ast}^-}\right)\\
    -2\lambda a_2\nu_0\left((q_{3*}^--\bar{q})e^{iu\xi_{1\ast}^-}\right)+3\lambda a_3\nu_0\left((q_{34}^--\bar{q})e^{iu\xi_{1\ast}^-}\right)\,.
\end{multline}
Using the definitions of $\phi_0^-$, $\psi_0^-$ and $\zeta_0^-$, and plugging everything inside \eqref{psi_cavity} one has:
\begin{multline}
    \psi_N(u)=-ua_2\phi_0^-(u)+\lambda a_2\partial_u\phi_0^-(u)+\lambda \psi_0^-(u)\left(a_1-a_2-2a_3\right)\\+\lambda \zeta_0^-(u)\left(
    3a_3-2a_2\right)+O\left(\frac{1}{\sqrt{N}}\right)\,.
\end{multline}

Finally, we simplify $\zeta_N$ too.
\begin{multline}\label{zeta_cavity}
    \zeta_N(u)=i\sqrt{N}\nu \left((\epsilon^{(2)}\epsilon^{(3)}-\bar{q} )e^{iu\xi_{1\ast}^-+iu\frac{\epsilon^{(1)}\epsilon^\ast}{\sqrt{N}}}\right)=\\=
    i\sqrt{N}\nu \left((\epsilon^{(2)}\epsilon^{(3)}-\bar{q} )e^{iu\xi_{1\ast}^-}
    \left(1+iu\frac{\epsilon^{(1)}\epsilon^\ast}{\sqrt{N}}\right)\right)+O\left(\frac{1}{\sqrt{N}}\right)=\\=i\sqrt{N}\nu \left((\epsilon^{(2)}\epsilon^{(3)}-\bar{q} ) e^{iu\xi_{1\ast}^-}\right)-u\,\nu \left(\epsilon^{(1)}\epsilon^\ast
    (\epsilon^{(2)}\epsilon^{(3)}-\bar{q} ) e^{iu\xi_{1\ast}^-}\right)+O\left(\frac{1}{\sqrt{N}}\right)\,.
\end{multline}
As usual, we first treat the second term with \eqref{main_estimate1}:
\begin{align}
    \nu \left(\epsilon^{(1)}\epsilon^\ast
    (\epsilon^{(2)}\epsilon^{(3)}-\bar{q} ) e^{iu\xi_{1\ast}^-}\right)=a_3\phi_0^-(u)+\eta
\end{align}
where $\eta=O(\nu(|R_{1*}^--\bar{q}|))=O(1/\sqrt{N})$.
The first term in the last line of \eqref{zeta_cavity} is instead treated with \eqref{main_estimate2}:
\begin{align}
    \nu_0\left((\epsilon^{(2)}\epsilon^{(3)}-\bar{q} ) e^{iu\xi_{1\ast}^-}\right)=\nu_0(\epsilon^{(2)}\epsilon^{(3)}-\bar{q} ) \nu_0\left(e^{iu\xi_{1\ast}^-}\right)+\nu_0'\left((\epsilon^{(2)}\epsilon^{(3)}-\bar{q} ) e^{iu\xi_{1\ast}^-}\right)+\eta'
\end{align}
The first term on the r.h.s. is still zero thanks to the Nishimori identities. We focus on the second one, using \eqref{derivata_nu} with $n=3$:
\begin{multline}
    \nu_0'\left((\epsilon^{(2)}\epsilon^{(3)}-\bar{q} ) e^{iu\xi_{1\ast}^-}\right)=\lambda a_2\nu_0\left((R_{12}^--\bar{q})e^{iu\xi_{1\ast}^-}\right)+\lambda a_2 \nu_0\left((R_{13}^--\bar{q})e^{iu\xi_{1\ast}^-}\right)\\
    +\lambda a_1\nu_0\left((R_{23}^--\bar{q})e^{iu\xi_{1\ast}^-}\right)-3\lambda a_3\nu_0\left((R_{14}^--\bar{q})e^{iu\xi_{1\ast}^-}\right)
    -3\lambda a_2 \nu_0\left((R_{24}^--\bar{q})e^{iu\xi_{1\ast}^-}\right)\\
    -3\lambda a_2\nu_0\left((R_{34}^--\bar{q})e^{iu\xi_{1\ast}^-}\right)+\lambda a_3\nu_0\left((R_{1*}^--\bar{q})e^{iu\xi_{1\ast}^-}\right)
    +\lambda a_2\nu_0\left((R_{2*}^--\bar{q})e^{iu\xi_{1\ast}^-}\right)\\+\lambda a_2\nu_0\left((R_{3*}^--\bar{q})e^{iu\xi_{1\ast}^-}\right)
    -3\lambda a_3\nu_0\left((R_{*4}^--\bar{q})e^{iu\xi_{1\ast}^-}\right)
    +6\lambda a_3\nu_0\left((R_{45}^--\bar{q})e^{iu\xi_{1\ast}^-}\right)\,.
\end{multline}
Using the definition of the $0$-time functions and plugging the previous inside \eqref{zeta_cavity} we get
\begin{multline}
    \zeta_N(u)=-ua_3\phi_0^-(u)+\lambda a_3\partial_u\phi_0^-(u)+\lambda \psi_0^-(u)\left( 4  a_2-6  a_3 \right)\\
    +\lambda\zeta_0^-(u)\left( a_1-6a_2+6a_3 \right)+O\left(\frac{1}{\sqrt{N}}\right)\,.
\end{multline}
With the computations above we have just shown that
\begin{multline}\label{sistem}
    \begin{pmatrix}
    \partial_u\phi_N\\
    \psi_N\\
    \zeta_N
    \end{pmatrix}=\lambda 
    \begin{pmatrix}
        a_1 & -2a_2 & a_3\\
        a_2 & a_1-a_2-2a_3 & 3a_3-2a_2\\
        a_3 & 4a_2-6a_3 & a_1-6a_2+6a_3\\
    \end{pmatrix}
     \begin{pmatrix}
    \partial_u\phi_0^-\\
    \psi_0^-\\
    \zeta_0^-
    \end{pmatrix}-u\phi_0^-(u)
     \begin{pmatrix}
    a_1\\
    a_2\\
    a_3
    \end{pmatrix}+O\left(\frac{1}{\sqrt{N}}\right)\,.
\end{multline}
We recall again that $O(1/\sqrt{N})$ has to be intended as a quantity whose norm can be bounded with $C_u/\sqrt{N}$ where $C_u$ has an at most quadratic dependence on $u$.

We claim that, up to $O\left(N^{-1/2}\right)$, one can replace in the r.h.s. of \eqref{sistem}
the quantities $\phi_0^-, \psi_0^-,$ and $\zeta_0^-$ by $\phi, \psi$ and $\zeta$ respectively. Consider for instance the difference between $\partial_u\phi$ and $\partial_u\phi_0^-$, the other terms can be treated analogously. By Lemma \ref{replacement_lemma} with $n=2$ one readily obtains that  $\partial_u\phi_N-\partial_u\phi_0^-=O\left(\dfrac{1}{\sqrt{N}}\right)$. Therefore equation \eqref{sistem} can be now rewritten in a more compact form as
\begin{align}\label{eq:system_compact}
    (\mathbbm{1}-\lambda M)
    \begin{pmatrix}
    \partial_u\phi_N\\
    \psi_N\\
    \zeta_N
    \end{pmatrix}=-u\phi_N(u)
     \begin{pmatrix}
    a_1\\
    a_2\\
    a_3
    \end{pmatrix}+O\left(\frac{1}{\sqrt{N}}\right)\,,
\end{align}
where $M$ is the $3\times3$ matrix in \eqref{sistem}.  In \cite{ElAlaoui} it is proved  that if  $\lambda$ is such that  $p(\lambda)$ is twice differentiable then  $\mathbbm{1}-\lambda M$ is invertible and in this case
\begin{align}\label{eq:final_covariances}
    \big[\mathbbm{1}-\lambda M\big]^{-1} \begin{pmatrix}
    a_1\\
    a_2\\
    a_3
    \end{pmatrix}=\begin{pmatrix}
        \frac{1}{\lambda}\left(-1+\frac{2}{1- \mu_2}+\frac{2}{1- \mu_1}+\frac{-3+3 \lambda a_1-2 \lambda  a_2}{\left(1- \mu_1\right)^2}\right)\\
        \frac{1}{\lambda }\left(\frac{-3+3 \lambda  a_1-2 \lambda  a_2}{\left(1- \mu_1\right)^2}+\frac{3}{1- \mu_2}\right)\\
        \frac{4 \lambda  a_2^2+(1-\lambda  a_1-5 \lambda  a_2) a_3+2 \lambda  a_3^2}{\left(1- \mu_1\right)^2\left(1- \mu_2\right)}
    \end{pmatrix}\,,
\end{align}
with $\mu_2\leq\mu_1$ and
\begin{align}
& \mu_1(\lambda)=\lambda(a_1-2 a_2+a_3) \\
& \mu_2(\lambda)=\lambda(a_1-3 a_2+2 a_3)\,.
\end{align}
In particular, when $p(\lambda)$ is twice differentiable one has $\mu_1(\lambda)<1$ ensuring the  existence of $\big[\mathbbm{1}-\lambda M\big]^{-1}$. Therefore the differential system \eqref{eq:system_compact} is equivalent to

\begin{align}\label{eq:pre_limit_version}
    \begin{pmatrix}
    \partial_u\phi_N\\
    \psi_N\\
    \zeta_N
    \end{pmatrix}=-u\phi_N(u)\big[\mathbbm{1}-\lambda M\big]^{-1}
     \begin{pmatrix}
    a_1\\
    a_2\\
    a_3
    \end{pmatrix}+O\left(\frac{1}{\sqrt{N}}\right)\,.
\end{align}
The first line of the previous system, together with the initial condition $\phi_N(0)=1$ is a Cauchy problem for $\phi_N$:
\begin{equation}\label{odephiN}
\begin{cases}
\partial_u\phi_N = - C \, u \,\phi_N(u)+ O\left(N^{-1/2}\right) ,\\
\phi_N(0)=1
\end{cases}
\end{equation}
where $C:=\sum_{j=1}^3\big[\mathbbm{1}-\lambda M\big]^{-1}_{1,j}a_j\,$. One can verify that $\phi_N$ is a sequence of absolutely continuous functions and then by Lemma \ref{lem:limit_ODE} one gets  $\phi_N\xrightarrow[]{N\to\infty}\phi$ where $\phi$ solves the same differential problem without the remainder. More precisely one obtains
\begin{align}\label{eq:cov_diagonal}
    \lim_{N\to\infty}\phi_N(u)=\phi(u)=\exp\Big(-\frac{u^2 C}{2}\Big)\,,\quad 
\end{align}
and then by Levi's continuity  Theorem we conclude that  $\xi_{1\ast}$ converges  in distribution to a  centered Gaussian with variance $C$.

\begin{remark} 
As a last check we note that $\partial_u^3\phi_N(u)$ is bounded for any fixed $u\in[-K,K]$ for instance, thanks to \eqref{eq:concentration_elalaoui}. Hence, we can exchange the second derivative w.r.t. $u$ with the $N$ limit (see e.g. \cite[Theorem 1.1]{ACid}). If we evaluate such derivative in $u=0$ we obtain the matching
\begin{align}
    \partial_u^2\phi_N(0)=-\nu(\xi_{1\ast}^2)\xrightarrow[]{N\to\infty}\partial_u^2\phi(0)=-C
\end{align}
which is in agreement with \cite{ElAlaoui}. Equivalently, one could have argued that thanks to \eqref{eq:concentration_elalaoui} the sequence $\xi_{1\ast}^2=N(q_{1\ast}-\bar{q})^2$ is uniformly integrable. Therefore, besides the convergence in distribution, one would have also convergence of moments up to the third for instance, that would be $0$.

The same argument holds true for the second derivatives of $\psi_N$ and $\zeta_N$. From the limiting version of \eqref{eq:pre_limit_version} it is clear that $\psi(u):=\lim_{N\to\infty}\psi_N(u)$ must satisfy
\begin{align}
    \psi(u)=-u\phi(u)\sum_{j=1}^3\big[\mathbbm{1}-\lambda M\big]^{-1}_{2,j}a_j
\end{align}
which in turn entails
\begin{align}\label{eq:cov_offdiagonalcommon}
    \partial_u\psi_N(0)=-\nu(\xi_{2\ast}\xi_{1\ast})\xrightarrow[]{N\to\infty}\partial_u\psi(0)=-\sum_{j=1}^3\big[\mathbbm{1}-\lambda M\big]_{2,j}^{-1}a_j\,.
\end{align}
Analogously:
\begin{align}\label{eq:cov_offdiagonal_different}
    \partial_u\zeta_N(0)=-\nu(\xi_{23}\xi_{1\ast})\xrightarrow[]{N\to\infty}\partial_u\zeta(0)=-\sum_{j=1}^3\big[\mathbbm{1}-\lambda M\big]_{3,j}^{-1}a_j\,.
\end{align}

\end{remark}

\section{Proof of Theorem \ref{th:main}}\label{sec:proofmain}
We now prove the main theorem with a generic number of replicas. We are going to prove that the vector $\boldsymbol{\xi}=(\xi_{ab})_{1\leq a<b\leq n}$ tends to a Gaussian vector with a suitable covariance matrix $\Sigma=(\Sigma_{ab,cd})$ where $1\leq a<b\leq n, 1\leq c<d\leq n$. Let us fix the notations:
\begin{align}\label{notations}
    &\mathbf{u}=(u_{ab})_{1\leq a<b\leq n}\,,\quad \boldsymbol{\xi}=(\xi_{ab})_{1\leq a<b\leq n}\,,\quad \mathbf{u}^\intercal\boldsymbol{\xi}=\sum_{1\leq a<b\leq n}u_{ab}\xi_{ab}\,.
\end{align}
Following the same procedure as before we now define the fields
\begin{align}
    &\phi_N(\mathbf{u})=\nu\left(e^{i\mathbf{u}^\intercal \boldsymbol{\xi}}\right)\\
    &\vec{\psi}_N(\mathbf{u})=\Big(i\nu\left(\xi_{r,n+1}e^{i\mathbf{u}^\intercal \boldsymbol{\xi}}\right)\Big)_{r=1,\dots,n}\\
    &\zeta_N(\mathbf{u})=i\nu\left(\xi_{n+1,n+2}e^{i\mathbf{u}^\intercal \boldsymbol{\xi}}\right)\,,
\end{align}
together with the global field
\begin{align}
    \boldsymbol{\Gamma}_N(\mathbf{u})=\begin{pmatrix}
        \partial_\mathbf{u}\phi_N(\mathbf{u})\\
        \vec{\psi}_N(\mathbf{u})\\
        \zeta_N(\mathbf{u})
    \end{pmatrix}=\begin{pmatrix}
        i\nu\left(\boldsymbol{\xi}e^{i\mathbf{u}^\intercal \boldsymbol{\xi}}\right)\\
        \Big(i\nu\left(\xi_{r,n+1}e^{i\mathbf{u}^\intercal \boldsymbol{\xi}}\right)\Big)_{r=1,\dots,n}\\
        i\nu\left(\xi_{n+1,n+2}e^{i\mathbf{u}^\intercal \boldsymbol{\xi}}\right)
    \end{pmatrix}\,.
\end{align}
We will also need to handle the fields without the last particle:
\begin{align}
    &\phi_0^-(\mathbf{u})=\nu_0\left(e^{i\mathbf{u}^\intercal \boldsymbol{\xi}^-}\right)\\
    &\vec{\psi}_0^-(\mathbf{u})=\Big(i\nu_0\left(\xi_{r,n+1}^-e^{i\mathbf{u}^\intercal \boldsymbol{\xi}^-}\right)\Big)_{r=1,\dots,n}\\
    &\zeta_0^-(\mathbf{u})=i\nu_0\left(\xi_{n+1,n+2}^-e^{i\mathbf{u}^\intercal \boldsymbol{\xi}^-}\right)\,,
\end{align}that will be collected in the vector $\boldsymbol{\Gamma}_0^-(\mathbf{u})$. Let us also introduce the following coefficients
\begin{align}
    a_{rs}^{ab}=\nu_0\Big((\epsilon^r\epsilon^s-\bar{q})\epsilon^a\epsilon^b\Big)\,.
\end{align}
Regardless of the values of $r\neq s$ and $a\neq b$, they can assume only three values
\begin{align}
    a_{rs}^{ab}=\begin{cases}
    a_1\,,\quad a=r\,,b=s\\
    a_2\,,\quad a=r\,,b\neq s\\
    a_3\,,\quad a\neq r\,,b\neq s
    \end{cases}\,.
\end{align}
We collect the simplifications of the cavity fields into the following
\begin{lemma}[Cavity computation]\label{lem:cavitycomp}
Under the hypothesis of Theorem \ref{th:main} we have
\begin{align}\label{eq:final_cavity}
    \boldsymbol{\Gamma}_N(\mathbf{u})=-\mathbf{A}\mathbf{u}\phi_0^-(\mathbf{u})+\lambda\mathbf{B}\boldsymbol{\Gamma}_0^-(\mathbf{u})+O(N^{-1/2})\,,
\end{align}
    where $\mathbf{A}=(a_{rs}^{ab})_{1\leq r<s\leq n+1\,\bigvee\,(r,s)=(n+1,n+2)}^{1\leq a<b\leq n}$ and, for $1\leq r<r'\leq n$ 
    \begin{align}\label{Breplicas1}
    &\mathbf{B}_{rr'}=(B^{ab}_{rr'})=\begin{cases}
    A_{rr'}^{ab}\,,\quad a<b\leq n\\
    -(n-1)A^{a,n+1}_{rr'}\,,\quad a\leq n,\, b=n+1\quad\\
    \frac{n(n-1)}{2}A_{rr'}^{n+1,n+2}\,,\quad a=n+1\,,b=n+2
    \end{cases}\,,\quad\\
    \label{Breplicas2}
    &\mathbf{B}_{r,n+1}=(B_{r,n+1}^{ab})=\begin{cases}
    a_{r,n+1}^{ab}\,,\quad a<b\leq n\\
    a_{r,n+1}^{a,n+1}-na_{r,n+1}^{a,n+2}\,,\quad a\leq n\,,b=n+1\\
    \frac{n(n+1)}{2}a_{r,n+1}^{n+2,n+3}-na_{r,n+1}^{n+1,n+2}\,,\quad a=n+1\,,b=n+2
    \end{cases}\,,\\
    \label{Breplicas3}
    &\mathbf{B}_{n+1,n+2}=(B_{n+1,n+2}^{ab})=
    \begin{cases}
    a_{n+1,n+2}^{ab}\,,\quad a<b\leq n\\
    2a_{n+1,n+2}^{a,n+1}-(n+1)a_{n+1,n+2}^{a,n+3}\,,\quad a\leq n\,,b=n+1\\
    \frac{(n+1)(n+2)}{2}a_{n+1,n+2}^{n+3,n+4}-2(n+1)a_{n+1,n+2}^{n+1,n+3}\,,\quad (a,b)=(n+1,n+2)
    \end{cases}\,.
\end{align}
\end{lemma}
The proof of Lemma \ref{lem:cavitycomp} is deferred to the Appendix. Thanks to Lemma \ref{replacement_lemma} we can replace $\phi_0^-$ with $\phi$ and the fields $\boldsymbol{\Gamma}_0^-$ with $\boldsymbol{\Gamma}$. Equation \eqref{eq:final_cavity} then rewrites as
\begin{align}
    (\mathbbm{1}-\lambda \mathbf{B})\boldsymbol{\Gamma}_N(\mathbf{u})=-\mathbf{A}\mathbf{u}\phi_N(\mathbf{u})+O(N^{-1/2})\,.
\end{align}
Theorem \ref{overlap_invarianceThm} grants us that if we move along the trajectories there characterized the overlap is constant, and so will be $\mathbf{B}$. Hence, we conclude that $\mathbbm{1}-\lambda\mathbf{B}$ is invertible for almost all $\lambda$ along these trajectories. In particular, if $\mathbbm{1}-\tilde\lambda \mathbf{B}$ is not invertible, there is another point on the $(\lambda,h)$ plane with a value $\lambda '$ arbitrarily close to $\tilde\lambda$ such that $\mathbbm{1}-\lambda' \mathbf{B}$ is invertible. For future convenience, we also extend the definition of the characteristic function to the plane $(\lambda.h)$ as follows:
\begin{align}
    \label{eq:extended_phi}
    \phi_N(\bu;\lambda,h)=\mathbb{E}_{\mathbf{z},\mathbf{h},\bx^*}\frac{\int\prod_{a=1}^n dP_X(\bx^{(a)})e^{-\sum_{a=1}^n\tilde H_N(\mathbf{x}^{(a)};\mathbf{x}^*,\mathbf{z},\mathbf{h})+i\bu^\intercal\boldsymbol{\xi}}}{\int\prod_{b=1}^n dP_X(\bx^{'(b)})e^{-\sum_{a=1}^n\tilde H_N(\mathbf{x}^{'(b)};\mathbf{x}^*,\mathbf{z},\mathbf{h})}}
\end{align}
with $\tilde H$ defined in \eqref{eq:Hamiltonian_complete} and $\xi_{ab}=\sqrt{N}(q_{ab}-\bar{q}(\lambda,h))$. We will make the dependency on $\lambda$ and $h$ explicit only when needed.

Suppose for the moment that $\lambda$ is such that we can invert $\mathbbm{1}-\lambda\mathbf{B}$ in the original problem (without external magnetic fields). We then rewrite the equation as
\begin{align}
    \begin{pmatrix}
        \partial_\mathbf{u}\phi_N(\mathbf{u})\\
        \vec{\psi}_N(\mathbf{u})\\
        \zeta_N(\mathbf{u})
    \end{pmatrix}=-(\mathbbm{1}-\lambda \mathbf{B})^{-1}\mathbf{A}\mathbf{u}\phi_N(\mathbf{u})+O(N^{-1/2})
\end{align}
whose first $n(n-1)/2$ components constitute a linear differential equation for $\phi$ with the initial condition $\phi(0)=1$. We recall again that the remainder has an at most quadratic dependence on $\mathbf{u}$. Hence, we can use Lemma \ref{lem:limit_ODE} with $f_N(s)=\phi_N(s\mathbf{u})$ to find a solution for $\phi(\mathbf{u}):=\lim_{N\to\infty}\phi_N(\mathbf{u})$. The solution of such equation must then be of the form
\begin{align}
    \phi(\mathbf{u})=\exp\Big(-\frac{1}{2}\mathbf{u}^\intercal \boldsymbol{\Sigma}\mathbf{u}\Big)
\end{align}where $\boldsymbol{\Sigma}$ is a matrix, whose elements depend on those of $(\mathbbm{1}-\lambda \mathbf{B})^{-1}\mathbf{A}$. 
Even when it is possible, inverting $\mathbbm{1}-\lambda\mathbf{B}$ is impractical, especially because our statement is for a generic number of replicas $n$. Nevertheless, thanks again to \cite[Theorem 1.1]{ACid}, we can exchange derivatives w.r.t. to any component of $\mathbf{u}$ in $\mathbf{u}=0$ (note that the control on each partial derivative is uniform w.r.t. the other coordinates thanks to the boundedness of the imaginary exponential). This allows us to match the elements of $\boldsymbol{\Sigma}$ with the following covariances:
\begin{align}\label{covariances}
\begin{split}
    &\Sigma_{ab,ab}=\sum_{j=1}^m\big[\mathbbm{1}-\lambda M\big]^{-1}_{1,j}a_j\,,\\
    &\Sigma_{ab,ac}=\sum_{j=1}^m\big[\mathbbm{1}-\lambda M\big]^{-1}_{2,j}a_j\,,\\
    &\Sigma_{ab,cd}=\sum_{j=1}^m\big[\mathbbm{1}-\lambda M\big]^{-1}_{3,j}a_j
\end{split}
\end{align}
with $a<b$, $c<d$, $a\neq c$, $b\neq d$, whose explicit form can be read in \eqref{eq:final_covariances}.

On the other hand, for those values $\tilde \lambda$ at which $\mathbbm{1}-\tilde\lambda \mathbf{B}$ is not invertible, we consider a point $(\lambda',h(\lambda'))$ as in Theorem \ref{overlap_invarianceThm} arbitrarily close to $(\tilde \lambda,0)$. For this point we can invert the matrix and proceed as before. Now, we note that the extended characteristic function \eqref{eq:extended_phi} is Lipschitz along the trajectories of Theorem \ref{overlap_invarianceThm}, and so will be its limit. A straightforward computation of the directional derivative indeed yields
\begin{align}
\begin{split}
    \frac{d}{d\lambda'} \phi_N(\mathbf{u};\lambda',h(\lambda'))&=(\partial_{\lambda'}-\bar{q}\partial_h)\phi_N(\mathbf{u};\lambda',h(\lambda'))\\
    &=O(1)+ \frac{1}{2}\mathbb{E}\Big\langle
    e^{i\mathbf{u}^\intercal\boldsymbol{\xi}}\Big[
    \sum_{0\leq a<b\leq n}\xi_{ab}^2-n\sum_{a=1}^n\xi^2_{a,n+1}+\frac{n(n-1)}{2}\xi_{n+1,n+2}^2
    \Big]
    \Big\rangle
\end{split}
\end{align}
where by $O(1)$ we mean a quantity that is bounded by a constant, depending only on the width of the support of $P_X$ and on the number of replicas $n$. Recall that $\bar{q}(\lambda',h(\lambda'))=\bar{q}(\lambda)$ for any $\lambda'\in[0,\lambda]$. Each of the other three contributions appearing above can be bounded thanks to Theorem \ref{thm:concalaoui}, leading to bounded derivative uniformly in $N$. Hence the non-invertibility points on the $\lambda$-axis can be completed by continuity.

We still need to make sure that $\phi_N(\mathbf{u};\tilde \lambda,0)\xrightarrow[]{N\to\infty}\phi(\mathbf{u};\tilde \lambda,0)$, the latter being well defined if $p$ is twice differentiable at $\tilde \lambda$ (recall that $\mathbbm{1}-\tilde\lambda M$ in \eqref{covariances} is thus invertible). This is clearly the case thanks to uniform (in $N$) Lipschitzness of the characteristic function. In fact:
\begin{align}
    \begin{split}
        &|\phi_N(\mathbf{u};\tilde \lambda,0)-\phi(\mathbf{u};\tilde \lambda,0)|=\\
        &|\phi_N(\mathbf{u};\tilde \lambda,0)-\phi_N(\mathbf{u};\lambda',h(\lambda'))+\phi_N(\mathbf{u};\lambda',h(\lambda'))-\phi(\mathbf{u};\lambda',h(\lambda'))+\phi(\mathbf{u};\lambda',h(\lambda'))-\phi(\mathbf{u};\tilde \lambda,0)|\\
        &\leq 2K|\tilde \lambda-\lambda'|+|\phi_N(\mathbf{u};\lambda',h(\lambda'))-\phi(\mathbf{u};\lambda',h(\lambda'))|
    \end{split}
\end{align}
where $K$ denotes the Lipschitz constant, both of $\phi_N$ and its limit $\phi$, and we simply used the triangular inequality repeatedly. Now we let $N\to\infty$ and then $\lambda'\to\lambda$ obtaining, and the proof is complete.

\section{Conclusions and perspectives}

In this paper we considered a class of disordered mean-field spin-glasses defined by a Hamiltonian with two body non centered Gaussian interaction and spin distribution with bounded support. The Hamiltonian is such that the model fulfills the Nishimori identities, and in this case the free energy can be expressed as the solution of a one dimensional variational principle. This is indeed a consequence of the Nishimori identities that force the order parameter, i.e. the overlap between two spins configurations, to be self-averaging everywhere but the critical points. We studied the fluctuations of the overlap around its mean and proved that the rescaled overlap vector converges, everywhere but the critical points, to a Gaussian vector in the thermodynamic limit for an arbitrary number of involved real replicas of the system. This result improves our knowledge of the statistical mechanics properties of such systems. {We notice that our techniques cannot be applied if the model is at the critical point. For Ising spins, for example, criticality occurs at $\lambda=1$, and in this case the quantities \eqref{eq:final_covariances}, as well as the elements of covariance of the Gaussian density in the CLT,  diverge. } Nevertheless, we believe the methods employed here are robust enough to generalize the result in various directions.  
 
A first generalization can be obtained replacing one dimensional spins with vectors spins. This corresponds to a finite rank matrix estimation problem in inference \cite{lesieur2015mmse,Lelarge2017FundamentalLO}. Here the order parameter becomes a "matrix overlap" and it is possible to show that Nishimori identities imply its concentration \cite{overlap_jean}. Hence it is reasonable to expect that the matrix overlap satisfies a Central Limit Theorem everywhere but at the critical points.

Recently, it was shown in \cite{Mukherjee2021FluctuationsOT,contucci2023limit} how to extend Ellis and Newman's results to the Curie-Weiss $p$-spin model. We believe that a similar extension can hold also in a wider context, i.e. for disordered models on the Nishimori line with higher order interactions. The corresponding inferential problem is called \emph{low rank tensor estimation}, which exhibits the same replica symmetric features of the standard low rank matrix estimation 
\cite{BarMacrLun20,Layered_Tensor_Jean_Miolane,Mourrat2}.

Another direction in which we believe one can extend the fluctuation analysis is for the class of mean-field multispecies models. They are characterized by an invariance under block-permutation among particles of different species. Their deterministic formulation, with the related fluctuations properties, were studied in \cite{Fedele10,Lowe}. In the disordered case with centered Gaussian couplings, the free energy was computed only under some technical convexity assumption \cite{MSK_original,panchenko_multi-SK,Dey2020FluctuationRF} or for spherical spins \cite{Baik,Bates,Subag}. Multispecies models on the Nishimori line instead were solved in
\cite{MSKNL,DBMNL} also in absence of convexity. For the latter, the central limit theorems for the rescaled overlap vector should be obtained with the same methods used in this work.

Finally, the free energy of spin glasses with rotationally invariant couplings, i.e. whose law is invariant under orthogonal transformation, was investigated in the high temperature phase \cite{bhattacharya2016high,fan2021replica}. In particular, the free energy fulfills a replica-symmetric variational principle, with the overlap as the order parameter. Moreover, using methods from Statistical Physics, it is possible to find a replica symmetric formula for the free entropy of inferential models with rotationally invariant noise, where the overlap appears again as an order parameter \cite{barbier2022bayes}. It would thus be interesting to study the fluctuation properties, both for the free energy and the overlap, in those cases.

\paragraph{Acknowledgments} The authors thank F. L. Toninelli for insightful discussions on \cite{quadratic_replica_coupling,GuerraToninelliCLT}. We also acknowledge J. Barbier, W. -K. Chen and A. El Alaoui for fruitful interactions. The work is partially supported by GNFM (Indam), EU H2020 ICT48 project Humane AI Net.

\small

\appendix

\section{Intermediate results for the proof of Theorem \ref{th:main}}
We begin with a modified version of Lemma \ref{lemma_dv}:
\begin{lemma}
For any bounded function $f$ of $n$ replicas (and not of the ground truth) the following holds:
\begin{multline}\label{derivata_nu_gen}
\frac{\mathrm{d}}{\mathrm{d} t} \nu_{t}(f)=\frac{\lambda}{2} \sum_{1 \leq l \neq l^{\prime} \leq n} \nu_{t}\left(\left(q_{l, l^{\prime}}^{-}-\bar{q}\right) \epsilon^{(l)} \epsilon^{\left(l^{\prime}\right)} f\right)-\lambda (n-1) \sum_{l=1}^{n} \nu_{t}\left((q_{l, n+1}^{-}-\bar{q}\right) \epsilon^{(l)} \epsilon^{(n+1)} f )\\
+\lambda \frac{n(n-1)}{2} \nu_{t}\left(\left(q_{n+1, n+2}^{-}-\bar{q}\right) \epsilon^{(n+1)} \epsilon^{(n+2)} f\right)\,.
\end{multline}
\end{lemma}
\begin{proof}
The proof follows from that of Lemma \eqref{lemma_dv} applying Nishimori identities, and uses the fact that $f$ does not depend on the ground truth.
\end{proof}

\subsection{Proof of Lemma \ref{lem:cavitycomp}}
Concerning the first $n(n-1)/2$ components of $\boldsymbol{\Gamma}_N$, given $1\leq r<r'\leq n$ we have 
\begin{multline}\label{partial_phi_gen}
\partial_{u_{rr'}}\phi_N(\mathbf{u})=i\nu\left(\xi_{rr'}e^{i\mathbf{u}^\intercal\boldsymbol{\xi}}\right)=i\sqrt{N}\nu\left((\epsilon^{r}\epsilon^{r'}-\bar{q})e^{i\mathbf{u}^\intercal\boldsymbol{\xi}}\right)=\\
    =i\sqrt{N}\nu\left((\epsilon^{r}\epsilon^{r'}-\bar{q})e^{i\mathbf{u}^\intercal\boldsymbol{\xi}^-+\frac{i}{\sqrt{N}}\sum_{a<b}^nu_{ab}\epsilon^a\epsilon^b}\right)=i\sqrt{N}\nu\left(
    (\epsilon^r\epsilon^{r'}-\bar{q})e^{i\mathbf{u}^\intercal\boldsymbol{\xi}^-}\right)\\
    -\sum_{a<b}^n u_{ab}\,\nu\left((\epsilon^r\epsilon^{r'}-\bar{q})\epsilon^a\epsilon^be^{i\mathbf{u}^\intercal\boldsymbol{\xi}^-}\right)+O\Big(\frac{1}{\sqrt{N}}\Big)\,.
\end{multline}
Here $O(1/\sqrt{N})$ can depend on $\mathbf{u}$ at most quadratically.
The last term can be simplified thanks to \eqref{main_estimate1}:
\begin{align}
    \sum_{a<b}^n u_{ab}\,\nu\left((\epsilon^r\epsilon^{r'}-\bar{q})\epsilon^a\epsilon^be^{i\mathbf{u}^\intercal\boldsymbol{\xi}^-}\right)=\left[\sum_{a<b}^nu_{ab}\,a^{ab}_{rr'}\right]\phi_0^-(\mathbf{u})+\delta\,.
\end{align}
Here $\delta=O(\nu(|q_{1\ast}^--\bar{q}|))=O(N^{-1/2})$ independently of $\mathbf{u}$. The term proportional to $\sqrt{N}$ needs instead the finer estimate \eqref{main_estimate2} with $\tau_1=1$, $\tau_2=\infty$, and consequently \eqref{derivata_nu_gen}:
\begin{multline}
    i\sqrt{N}\nu\left(
    (\epsilon^r\epsilon^{r'}-\bar{q})e^{i\mathbf{u}^\intercal\boldsymbol{\xi}^-}\right)=i\sqrt{N}\Big[
    \lambda\sum_{a<b}^n\nu_0\Big( (\epsilon^r\epsilon^{r'}-\bar{q})\epsilon^a\epsilon^b(q_{ab}^--\bar{q}) e^{i\mathbf{u}^\intercal\boldsymbol{\xi}^-}\Big)\\
    -\lambda(n-1)\sum_{a=1}^n\nu_0\Big((\epsilon^r\epsilon^{r'}-\bar{q})\epsilon^a\epsilon^{n+1}(q_{a,n+1}^--\bar{q}) e^{i\mathbf{u}^\intercal\boldsymbol{\xi}^-}\Big)\\
    +\frac{\lambda n(n-1)}{2}\nu_0\Big((\epsilon^r\epsilon^{r'}-\bar{q})\epsilon^{n+1}\epsilon^{n+2}(q_{n+1,n+2}^--\bar{q}) e^{i\mathbf{u}^\intercal\boldsymbol{\xi}^-}\Big)
    \Big]+\delta'\,,
\end{multline}
where $\delta'=\sqrt{N}O(\nu((q_{1\ast}^--\bar{q})^2))=O(N^{-1/2})$. Now consider that at $t=0$ the $\epsilon$ particle decouples from the other $N-1$-s. Hence we recover the coefficients $a_{rr'}^{ab}$. For future convenience we gather them into a rectangular matrix $\mathbf{A}=(a_{rs}^{ab})_{1\leq r<s\leq n+1\,\bigvee\,(r,s)=(n+1,n+2)}^{1\leq a<b\leq n}$.
Plugging everything back into \eqref{partial_phi_gen} we get
\begin{align}
    \partial_{u_{rr'}}\phi_N(\mathbf{u})=-(\mathbf{A}\mathbf{u})_{rr'}\phi_0^-(\mathbf{u})+\lambda \mathbf{B}_{rr'}\cdot\boldsymbol{\Gamma}_0^-(\mathbf{u})+O\Big(\frac{1}{\sqrt{N}}\Big)
\end{align}
with $\mathbf{B}_{rr'}$ as in \eqref{Breplicas1}, and a remainder that can depend at most quadratically in $\mathbf{u}$.

From this moment on, every remainder $O(N^{-1/2})$ depends at most quadratically on $\mathbf{u}$ if not specified. Let us turn to $\vec{\psi}$:
\begin{multline}
    \psi_{N,r}(\mathbf{u})=i\sqrt{N}\nu\Big((\epsilon^{r}\epsilon^{n+1}-\bar{q})e^{i\mathbf{u}^\intercal\boldsymbol{\xi}}\Big)=i\sqrt{N}\nu\Big((\epsilon^{r}\epsilon^{n+1}-\bar{q})e^{i\mathbf{u}^\intercal\boldsymbol{\xi}^-}\Big)\\-\sum_{a<b}^nu_{ab}\nu\Big((\epsilon^{r}\epsilon^{n+1}-\bar{q})\epsilon^{a}\epsilon^{b} e^{i\mathbf{u}^\intercal\boldsymbol{\xi}^-}
    \Big)+O(N^{-1/2})\,.
\end{multline}
The last term can be approximated by means of \eqref{main_estimate1}:
\begin{multline}
    \sum_{a<b}^nu_{ab}\nu\Big((\epsilon^{r}\epsilon^{n+1}-\bar{q})\epsilon^{a}\epsilon^{b} e^{i\mathbf{u}^\intercal\boldsymbol{\xi}^-}
    \Big)=\sum_{a<b}^nu_{ab}\nu_0\Big((\epsilon^{r}\epsilon^{n+1}-\bar{q})\epsilon^{a}\epsilon^{b} e^{i\mathbf{u}^\intercal\boldsymbol{\xi}^-}
    \Big)+O(N^{-1/2})=\\=\sum_{a<b}^nu_{ab}\,a^{ab}_{r,n+1}
    \phi_0^-(\mathbf{u})+O(N^{-1/2})=(\mathbf{A}\mathbf{u})_{r,n+1} \phi_0^-(\mathbf{u})+O(N^{-1/2})\,.
\end{multline}
The first term needs instead the finer estimate \eqref{main_estimate2} and \eqref{derivata_nu_gen}:
\begin{multline}
    i\sqrt{N}\nu\Big((\epsilon^{r}\epsilon^{n+1}-\bar{q})e^{i\mathbf{u}^\intercal\boldsymbol{\xi}^-}\Big)=i\sqrt{N}\Big[
    \lambda\sum_{a<b}^{n+1}\nu_0\Big((\epsilon^{r}\epsilon^{n+1}-\bar{q})\epsilon^a\epsilon^b(q_{ab}^--\bar{q}) e^{i\mathbf{u}^\intercal\boldsymbol{\xi}^-}\Big)\\
    -\lambda n\sum_{a=1}^{n+1}\nu_0\Big((\epsilon^{r}\epsilon^{n+1}-\bar{q})\epsilon^{a}\epsilon^{n+2}(q_{a,n+2}^--\bar{q}) e^{i\mathbf{u}^\intercal\boldsymbol{\xi}^-}\Big)\\+\frac{\lambda n(n+1)}{2}\nu_0\Big((\epsilon^{r}\epsilon^{n+1}-\bar{q})\epsilon^{n+2}\epsilon^{n+3}(q_{n+2,n+3}^--\bar{q}) e^{i\mathbf{u}^\intercal\boldsymbol{\xi}^-}\Big)
    \Big]+O(N^{-1/2})\,.
\end{multline}
We treat the three contributions in the three lines of the previous equation separately, keeping in mind that under $\nu_0$ the $\epsilon$-particle is factorized out. Concerning the first, by splitting the sum we have
\begin{align}
    i\sqrt{N}\sum_{a<b}^{n+1}\nu_0\Big((\epsilon^{r}\epsilon^{n+1}-\bar{q})\epsilon^a\epsilon^b(q_{ab}^--\bar{q}) e^{i\mathbf{u}^\intercal\boldsymbol{\xi}^-}\Big)=\sum_{a<b}^n a_{r,n+1}^{ab}\partial_{u_{ab}}\phi_0^-(\mathbf{u})+\sum_{a=1}^n a_{r,n+1}^{a,n+1}\psi_{0,a}^-(\mathbf{u})\,.
\end{align}
The second contribution in the second line yields instead
\begin{align}
    i\sqrt{N}\sum_{a=1}^{n+1}\nu_0\Big((\epsilon^{r}\epsilon^{n+1}-\bar{q})\epsilon^{a}\epsilon^{n+2}(q_{a,n+2}^--\bar{q}) e^{i\mathbf{u}^\intercal\boldsymbol{\xi}^-}\Big)=\sum_{a=1}^n a_{r,n+1}^{a,n+2}\psi_{0,a}^-(\mathbf{u})+a_{r,n+1}^{n+1,n+2}\zeta_0^-(\mathbf{u})\,.
\end{align}
Finally:
\begin{align}
    \nu_0\Big((\epsilon^{r}\epsilon^{n+1}-\bar{q})\epsilon^{n+2}\epsilon^{n+3}(q_{n+2,n+3}^--\bar{q}) e^{i\mathbf{u}^\intercal\boldsymbol{\xi}^-}\Big)=a_{r,n+1}^{n+2,n+3}\zeta_0^-(\mathbf{u})\,.
\end{align}
Putting everything together we get a cavity approximation for $\vec{\psi}$:
\begin{align}
    \psi_{N,r}(\mathbf{u})=-(\mathbf{A}\mathbf{u})_{r,n+1}\,\phi_0^-(\mathbf{u})+\lambda \mathbf{B}_{r,n+1}\cdot\boldsymbol{\Gamma}_0^-(\mathbf{u})+O(N^{-1/2})
\end{align}
with $\mathbf{B}_{r,n+1}$ as in \eqref{Breplicas2}.

We now finally turn to
\begin{multline}
    \zeta_N(\mathbf{u})=i\sqrt{N}\nu\Big((\epsilon^{n+1}\epsilon^{n+2}-\bar{q})e^{i\mathbf{u}^\intercal\boldsymbol{\xi}}\Big)=i\sqrt{N}\nu\Big((\epsilon^{n+1}\epsilon^{n+2}-\bar{q})e^{i\mathbf{u}^\intercal\boldsymbol{\xi}^-}\Big)\\
    -\sum_{a<b}^nu_{ab}\nu\Big(\epsilon^a\epsilon^b(\epsilon^{n+1}\epsilon^{n+2}-\bar{q})e^{i\mathbf{u}^\intercal\boldsymbol{\xi}^-}\Big)+O(N^{-1/2})\,.
\end{multline}
As before, we first focus on the last term
\begin{align}
    \nu\Big(\epsilon^a\epsilon^b(\epsilon^{n+1}\epsilon^{n+2}-\bar{q})e^{i\mathbf{u}^\intercal\boldsymbol{\xi}^-}\Big)=\sum_{a<b}u_{ab}a^{ab}_{n+1,n+2}\,\phi_0^-(\mathbf{u})+O(N^{-1/2})
\end{align}
where we used \eqref{main_estimate1}. Secondly:
\begin{multline}
    i\sqrt{N}\nu\Big((\epsilon^{n+1}\epsilon^{n+2}-\bar{q})e^{i\mathbf{u}^\intercal\boldsymbol{\xi}^-}\Big)
    =i\sqrt{N}\Big[\lambda\sum_{a<b}^{n+2}\nu_0\Big((\epsilon^{n+1}\epsilon^{n+2}-\bar{q})\epsilon^a\epsilon^b(q_{ab}^--\bar{q}) e^{i\mathbf{u}^\intercal\boldsymbol{\xi}^-}\Big)\\
    -\lambda (n+1)\sum_{a=1}^{n+2}\nu_0\Big((\epsilon^{n+1}\epsilon^{n+2}-\bar{q})\epsilon^a\epsilon^{n+3}(q_{a,n+3}^--\bar{q}) e^{i\mathbf{u}^\intercal\boldsymbol{\xi}^-}\Big)\\
    +\frac{\lambda (n+2)(n+1)}{2}\nu_0\Big((\epsilon^{n+1}\epsilon^{n+2}-\bar{q})\epsilon^{n+3}\epsilon^{n+4}(q_{n+3,n+4}^--\bar{q}) e^{i\mathbf{u}^\intercal\boldsymbol{\xi}^-}\Big)\Big]+O(N^{-1/2}),.
\end{multline}
The contribution in the first line yields
\begin{multline}
    i\sqrt{N}\sum_{a<b}^{n+2}\nu_0\Big((\epsilon^{n+1}\epsilon^{n+2}-\bar{q})\epsilon^a\epsilon^b(q_{ab}^--\bar{q}) e^{i\mathbf{u}^\intercal\boldsymbol{\xi}^-}\Big)=
    \sum_{a<b}^na^{ab}_{n+1,n+2}\partial_{u_{ab}}\phi_0^-(\mathbf{u})\\
    +\sum_{a=1}^{n}a_{n+1,n+2}^{a,n+1}\psi_{0,a}^-(\mathbf{u})
    +\sum_{a=1}^{n}a_{n+1,n+2}^{a,n+2}\psi_{0,a}^-(\mathbf{u})
    +a_{n+1,n+2}^{n+1,n+2}\zeta_0^-(\mathbf{u})=\\
    =\sum_{a<b}^na^{ab}_{n+1,n+2}\partial_{u_{ab}}\phi_0^-(\mathbf{u})
    +2\sum_{a=1}^{n}a_{n+1,n+2}^{a,n+1}\psi_{0,a}^-(\mathbf{u})+a_{n+1,n+2}^{n+1,n+2}\zeta_0^-(\mathbf{u})\,.
\end{multline}
The contribution in the second line is
\begin{align}
    i\sqrt{N}\sum_{a=1}^{n+2}\nu_0\Big((\epsilon^{n+1}\epsilon^{n+2}-\bar{q})\epsilon^a\epsilon^{n+3}(q_{a,n+3}^--\bar{q}) e^{i\mathbf{u}^\intercal\boldsymbol{\xi}^-}\Big)
    =\sum_{a=1}^na_{n+1,n+2}^{a,n+3}\psi_{0,a}^-(\mathbf{u})+2a_{n+1,n+2}^{n+1,n+3}\zeta_0^-(\mathbf{u})\,.
\end{align}
The third line contributes just with
\begin{align}
    i\sqrt{N}\nu_0\Big((\epsilon^{n+1}\epsilon^{n+2}-\bar{q})\epsilon^{n+3}\epsilon^{n+4}(q_{n+3,n+4}^--\bar{q}) e^{i\mathbf{u}^\intercal\boldsymbol{\xi}^-}\Big)=a_{n+1,n+2}^{n+3,n+4}\zeta_0^-(\mathbf{u})\,.
\end{align}
For the last time, gathering all the contributions we have
\begin{align}
    \zeta_N(\mathbf{u})=-(\mathbf{A}\mathbf{u})_{n+1,n+2}\phi_0^-(\mathbf{u})+\lambda\mathbf{B}_{n+1,n+2}\cdot\boldsymbol{\Gamma}_0^-(\mathbf{u})+O(N^{-1/2})
\end{align}
with
\begin{align}
    \mathbf{B}_{n+1,n+2}=(B_{n+1,n+2}^{ab})=
    \begin{cases}
    a_{n+1,n+2}^{ab}\,,\quad a<b\leq n\\
    2a_{n+1,n+2}^{a,n+1}-(n+1)a_{n+1,n+2}^{a,n+3}\,,\quad a\leq n\,,b=n+1\\
    \frac{(n+1)(n+2)}{2}a_{n+1,n+2}^{n+3,n+4}-2(n+1)a_{n+1,n+2}^{n+1,n+3}\,,\quad (a,b)=(n+1,n+2)
    \end{cases}\,.
\end{align}
Hence, we have just proven that the following linear system of equations holds
\begin{align}
    \boldsymbol{\Gamma}_N(\mathbf{u})=-\mathbf{A}\mathbf{u}\phi_0^-(\mathbf{u})+\lambda\mathbf{B}\boldsymbol{\Gamma}_0^-(\mathbf{u})+O(N^{-1/2})\,.
\end{align}

\printbibliography

@book{contucci_giardina_2012, place={Cambridge},
title={Perspectives on Spin Glasses}, DOI={10.1017/CBO9781139049306},
publisher={Cambridge University Press}, author={Contucci, Pierluigi and Giardinà, Cristian},
year={2012}
}

@book{2015sherrinpanchenkogton,
title={The Sherrington-Kirkpatrick Model},
author={Panchenko, Dmitry},
year={2015},
publisher={Springer},
}

@book{Tala_vol1,
title={Mean Field Models for Spin Glasses: Volume I: Basic Examples},
author={Talagrand, Michel},
year={2010},
publisher={Springer}
}

@book{Tala_vol2,
author = {Talagrand, Michel},
year = {2011},
month = {01},
pages = {},
publisher={Springer},
title = {Mean Field Models for Spin Glasses: Volume II: Advanced Replica-Symmetry and Low Temperature},
volume = {55},
isbn = {978-3-642-22252-8},
journal = {Ergebnisse der Mathematik und ihrer Grenzgebiete},
doi = {10.1007/978-3-642-22253-5}
}

@article{adaptive, 
title = {{The adaptive interpolation method: a simple scheme to prove replica formulas in Bayesian inference}}, 
author = {Barbier, Jean and Macris, Nicolas}, 
journal = {Probability Theory and Related Fields}, 
volume = {174}, 
year = {2019}
}

@ARTICLE{Mourrat2,
       author = {{Chen}, Hong-Bin and {Mourrat}, Jean-Christophe and {Xia}, Jiaming},
        title = "{Statistical inference of finite-rank tensors}",
      journal = {arXiv e-prints},
         year = 2021,
archivePrefix = {arXiv},
       eprint = {2104.05360},
}

@article{Nishi_id_PC,
author = {Contucci, Pierluigi and Morita, Satoshi and Nishimori, Hidetoshi},
year = {2005},
title = {Surface Terms on the Nishimori Line of the Gaussian Edwards-Anderson Model},
volume = {122},
journal = {Journal of Statistical Physics},
}

@ARTICLE{contucci_morita_nishimori,
author = {{Morita}, Satoshi and {Nishimori}, Hidetoshi and {Contucci}, Pierluigi},
title = {Griffiths inequalities for the Gaussian spin glass},
journal = {Journal of Physics A Mathematical General},
year = 2004,
volume = {37},
}

@Article{Camilli_mismatch,
	title={{An inference problem in a mismatched setting: a spin-glass model with  Mattis interaction}},
	author={Francesco Camilli and Pierluigi Contucci and Emanuele Mingione},
	journal={SciPost Phys.},
	volume={12},
	issue={4},
	pages={125},
	year={2022},
	publisher={SciPost},
	doi={10.21468/SciPostPhys.12.4.125},
}

@article{baffioni2000some,
  title={Some exact results on the ultrametric overlap distribution in mean field spin glass models (I)},
  author={Baffioni, Francesco and Rosati, Francesco},
  journal={The European Physical Journal B-Condensed Matter and Complex Systems},
  volume={17},
  pages={439--447},
  year={2000},
  publisher={Springer}
}

@article{MSKNL,
  title={The multi-species mean-field spin-glass on the Nishimori line},
  author={Alberici, Diego and Camilli, Francesco and Contucci, Pierluigi and Mingione, Emanuele},
  journal={Journal of Statistical Physics},
  volume={182},
  number={1},
  pages={1--20},
  year={2021},
  publisher={Springer}
}

@ARTICLE{MSK_original,
author = {Barra, Adriano and Contucci, Pierluigi and Mingione, Emanuele and Tantari, Daniele},
year = {2013},
title = {Multi-Species Mean Field Spin Glasses. Rigorous Results},
volume = {16},
journal = {Annales Institut Henri Poincaré},
}

@book{nishimori01,
  address = {Oxford; New York},
  author = {Nishimori, Hidetoshi},
  publisher = {Oxford University Press},
  title = {Statistical Physics of Spin Glasses and Information Processing: an Introduction},
  year = 2001
}

@article{GG_contucci_Giardina,
author = {Contucci, Pierluigi and Giardina, Cristian},
year = {2005},
title = {The Ghirlanda-Guerra Identities},
volume = {126},
journal = {Journal of Statistical Physics},
}

@ARTICLE{GG_original,
	doi = {10.1088/0305-4470/31/46/006},
	url = {https://doi.org/10.1088/0305-4470/31/46/006},
	year = 1998,
	publisher = {{IOP} Publishing},
	volume = {31},
	number = {46},
	pages = {9149--9155},
	author = {Stefano Ghirlanda and Francesco Guerra},
	title = {General properties of overlap probability distributions in disordered spin systems. Towards Parisi ultrametricity},
	journal = {Journal of Physics A: Mathematical and General},
}

@book{ellis_book_largedev,
title = {Entropy, Large Deviations, and Statistical Mechanics},
author = {Ellis, Richard},
publisher = {Springer},
year =      {2006},
}

@ARTICLE{interp_guerra_2002,
author = {{Guerra}, Francesco and {Toninelli}, Fabio Lucio},
title = "{The Thermodynamic Limit in Mean Field Spin Glass Models}",
journal = {Communications in Mathematical Physics},
year = 2002,
volume = {230},
}

@ARTICLE{ChenATsharp,
       author = {{Chen}, Wei-Kuo},
        title = "{On the Almeida-Thouless transition line in the SK model with centered Gaussian external field}",
      journal = {arXiv e-prints},
         year = 2021,
archivePrefix = {arXiv},
       eprint = {2103.04802},
 primaryClass = {cond-mat.dis-nn}
}

@Article{Guerra_upper_bound,
author="Guerra, Francesco",
title="Broken Replica Symmetry Bounds in the Mean Field Spin Glass Model",
journal="Communications in Mathematical Physics",
year="2003",
volume="233",
}

@article{panchenko_multi-SK,
author = "Panchenko, Dmitry",
journal = "Annals of Probability",
title = "The free energy in a multi-species Sherrington–Kirkpatrick model",
volume = "43",
year = "2015"
}

@article{ASS,
  title={Extended variational principle for the Sherrington-Kirkpatrick spin-glass model},
  author={Aizenman, Michael and Sims, Robert and Starr, Shannon L},
  journal={Physical Review B},
  volume={68},
  number={21},
  pages={214403},
  year={2003},
  publisher={APS}
}

@ARTICLE{DBMNL,
       author = {{Alberici}, Diego and {Camilli}, Francesco and {Contucci}, Pierluigi and {Mingione}, Emanuele},
        title = "{The solution of the deep Boltzmann machine on the Nishimori line}",
      journal = {Communications in Mathematical Physics },
         year = 2021,
        month = {07},
          doi = {10.1007/s00220-021-04165-0},
}

@inproceedings{el2018estimation,
  author = {El Alaoui, Ahmed and Krzakala, Florent},
  booktitle = {IEEE International Symposium on Information Theory (ISIT)},
  organization = {IEEE},
  pages = {1874--1878},
  title = {Estimation in the spiked Wigner model: a short proof of the replica formula},
  year = {2018}}

@article{fan2021replica,
  title={The replica-symmetric free energy for Ising spin glasses with orthogonally invariant couplings},
  author={Fan, Zhou and Wu, Yihong},
  journal={arXiv preprint arXiv:2105.02797},
  year={2021}
}

@article{overlap_jean,
    author = {Barbier, Jean},
    title = "{Overlap matrix concentration in optimal Bayesian inference}",
    journal = {Information and Inference: A Journal of the IMA},
    volume = {10},
    number = {2},
    pages = {597-623},
    year = {2020},
    month = {05},
    issn = {2049-8772},
    doi = {10.1093/imaiai/iaaa008},
}

@article{ACid,
author = {Aizenman, Michael and Contucci, Pierluigi},
year = {1998},
month = {01},
pages = {},
title = {On the Stability of the Quenched State in Mean Field Spin Glass Models},
volume = {92},
journal = {Journal of Statistical Physics},
doi = {10.1023/A:1023080223894}
}

@article{bhattacharya2016high,
  title={High Temperature Asymptotics of Orthogonal Mean-Field Spin Glasses},
  author={Bhattacharya, Bhaswar B. and Sen, Subhabrata},
  journal={Journal of Statistical Physics},
  volume={162},
  number={1},
  pages={63--80},
  year={2016},
  publisher={Springer}
}

@article{Tala_parisiformula,
 ISSN = {0003486X},
 doi = {10.4007/annals.2006.163.221},
 author = {Michel Talagrand},
 journal = {Annals of Mathematics},
 number = {1},
 pages = {221--263},
 publisher = {Annals of Mathematics},
 title = {The Parisi Formula},
 volume = {163},
 year = {2006}
}

@article{Lelarge2017FundamentalLO,
  title={Fundamental limits of symmetric low-rank matrix estimation},
  author={M. Lelarge and L{\'e}o Miolane},
  journal={Probability Theory and Related Fields},
  year={2017},
  volume={173},
  pages={859-929}
}

@article{Barbier_2019,
	doi = {10.1088/1751-8121/ab2735},
	url = {https://doi.org/10.1088/1751-8121/ab2735},
	year = 2019,
	publisher = {{IOP} Publishing},
	volume = {52},
	number = {29},
	pages = {294002},
	author = {Jean Barbier and Nicolas Macris},
	title = {The adaptive interpolation method for proving replica formulas. Applications to the Curie{\textendash}Weiss and Wigner spike models},
	journal = {Journal of Physics A: Mathematical and Theoretical},
}

@book{MPV,
author = {Mezard, M and Parisi, G and Virasoro, M},
title = {Spin Glass Theory and Beyond},
publisher = {WORLD SCIENTIFIC},
year = {1986},
doi = {10.1142/0271},
address = {},
edition   = {},
}

@inproceedings{lesieur2015mmse,
  title={MMSE of probabilistic low-rank matrix estimation: Universality with respect to the output channel},
  author={Lesieur, Thibault and Krzakala, Florent and Zdeborov{\'a}, Lenka},
  booktitle={2015 53rd Annual Allerton Conference on Communication, Control, and Computing (Allerton)},
  pages={680--687},
  year={2015},
  organization={IEEE}
}

@article{Johnstone_WSM,
author = {Iain M. Johnstone},
title = {{On the distribution of the largest eigenvalue in principal components analysis}},
volume = {29},
journal = {The Annals of Statistics},
number = {2},
publisher = {Institute of Mathematical Statistics},
pages = {295 -- 327},
year = {2001},
doi = {10.1214/aos/1009210544},
}

@book{evans10,
  address = {Providence, R.I.},
  author = {Evans, Lawrence C.},
  isbn = {9780821849743},
  publisher = {American Mathematical Society},
  title = {Partial differential equations},
  year = 2010
}

@book{mezard2009information,
  title={Information, physics, and computation},
  author={M\'ezard, Marc and Montanari, Andrea},
  year={2009},
  publisher={Oxford University Press}
}

@article{COJAOGHLANcavity,
title = {Information-theoretic thresholds from the cavity method},
journal = {Advances in Mathematics},
volume = {333},
pages = {694-795},
year = {2018},
issn = {0001-8708},
doi = {https://doi.org/10.1016/j.aim.2018.05.029},
author = {Amin Coja-Oghlan and Florent Krzakala and Will Perkins and Lenka Zdeborová},
}

@article{ElAlaoui,
author = {Ahmed El Alaoui and Florent Krzakala and Michael Jordan},
title = {{Fundamental limits of detection in the spiked Wigner model}},
volume = {48},
journal = {The Annals of Statistics},
number = {2},
publisher = {Institute of Mathematical Statistics},
pages = {863 -- 885},
year = {2020},
doi = {10.1214/19-AOS1826}
}

@article{barbier2022bayes,
  title={Bayes-optimal limits in structured PCA, and how to reach them},
  author={Barbier, Jean and Camilli, Francesco and Mondelli, Marco and Saenz, Manuel},
  journal={arXiv preprint arXiv:2210.01237},
  year={2022}
}

@article{quadratic_replica_coupling,
author = {Guerra, Francesco and Toninelli, Fabio},
year = {2002},
month = {01},
pages = {},
title = {Quadratic replica coupling in the Sherrington-Kirkpatrick mean field spin glass model},
volume = {43},
journal = {Journal of Mathematical Physics},
doi = {10.1063/1.1483378}
}

@ARTICLE{EllisNewman78,
       author = {{Ellis}, Richard S. and {Newman}, Charles M.},
        title = "{The statistics of Curie-Weiss models}",
      journal = {Journal of Statistical Physics},
         year = 1978,
        month = aug,
       volume = {19},
       number = {2},
        pages = {149-161},
          doi = {10.1007/BF01012508},
       adsurl = {https://ui.adsabs.harvard.edu/abs/1978JSP....19..149E}
}

@article{Mukherjee2021FluctuationsOT,
  title={Fluctuations of the Magnetization in the p-Spin Curie–Weiss Model},
  author={Somabha Mukherjee and Jaesung Son and Bhaswar B. Bhattacharya},
  journal={Communications in Mathematical Physics},
  year={2021},
  volume={387},
  pages={681 - 728}
}

@article{contucci2023limit,
  title={Limit theorems for the cubic mean-field Ising model},
  author={Contucci, Pierluigi and Mingione, Emanuele and Osabutey, Godwin},
  journal={arXiv preprint arXiv:2303.14578},
  year={2023}
}

@INPROCEEDINGS{Layered_Tensor_Jean_Miolane,
  author={Barbier, Jean and Macris, Nicolas and Miolane, Léo},
  booktitle={2017 55th Annual Allerton Conference on Communication, Control, and Computing (Allerton)}, 
  title={The layered structure of tensor estimation and its mutual information}, 
  year={2017},
  volume={},
  number={},
  pages={1056-1063},
  doi={10.1109/ALLERTON.2017.8262854}}

@ARTICLE{GuerraToninelliCLT,
       author = {{Guerra}, Francesco and {Lucio Toninelli}, Fabio},
        title = "{Central limit theorem for fluctuations in the high temperature region of the Sherrington{\textendash}Kirkpatrick spin glass model}",
      journal = {Journal of Mathematical Physics},
         year = 2002,
        month = dec,
       volume = {43},
       number = {12},
        pages = {6224-6237},
          doi = {10.1063/1.1515109},
archivePrefix = {arXiv},
       eprint = {cond-mat/0201092},
 primaryClass = {cond-mat.dis-nn},
       adsurl = {https://ui.adsabs.harvard.edu/abs/2002JMP....43.6224G}
}

@article{ALR88,
author = {Aizenman, Michael and Lebowitz, Joel and Ruelle, D.},
year = {1988},
month = {09},
pages = {527-527},
title = {Some rigorous results on the Sherrington-Kirkpatrick spin glass model},
volume = {116},
journal = {Communications in Mathematical Physics},
doi = {10.1007/BF01229207}
}

@article{Chenferromagnetic14,
     author = {Chen, Wei-Kuo},
     title = {On the mixed even-spin {Sherrington-Kirkpatrick} model with ferromagnetic interaction},
     journal = {Annales de l'I.H.P. Probabilit\'es et statistiques},
     pages = {63--83},
     publisher = {Gauthier-Villars},
     volume = {50},
     number = {1},
     year = {2014},
     doi = {10.1214/12-AIHP521},
     zbl = {1290.60101},
     mrnumber = {3161522},
     language = {en},
     url = {http://www.numdam.org/articles/10.1214/12-AIHP521/}
}

@article{Panchenkoultra,
author = {Panchenko, Dmitry},
year = {2011},
month = {12},
pages = {},
title = {The Parisi ultrametricity conjecture},
volume = {177},
journal = {Annals of Mathematics},
doi = {10.4007/annals.2013.177.1.8}
}

@article{PanWeDey,
author = {Chen, Wei-Kuo and Dey, Partha and Panchenko, Dmitry},
year = {2017},
month = {06},
pages = {},
title = {Fluctuations of the free energy in the mixed p-spin models with external field},
volume = {168},
journal = {Probability Theory and Related Fields},
doi = {10.1007/s00440-016-0705-5}
}

@book{Chattebook,
title={Superconcentration and Related Topics},
author={Chatterjee, Sourav},
year={2010},
publisher={Springer}
}

@article{Fedele10,
author = {Fedele, Micaela and Contucci, Pierluigi},
year = {2010},
month = {11},
pages = {1186-1205},
title = {Scaling Limits for Multi-species Statistical Mechanics Mean-Field Models},
volume = {144},
journal = {Journal of Statistical Physics},
doi = {10.1007/s10955-011-0334-4}
}

@article{Lowe,
author = {Knöpfel, Holger and Löwe, Matthias and Schubert, Kristina and Sinulis, Arthur},
year = {2020},
month = {03},
pages = {},
title = {Fluctuation Results for General Block Spin Ising Models},
volume = {178},
journal = {Journal of Statistical Physics},
doi = {10.1007/s10955-020-02489-0}
}

@ARTICLE{Pastur91,
       author = {{Pastur}, L.~A. and {Shcherbina}, M.~V.},
        title = "{Absence of self-averaging of the order parameter in the Sherrington-Kirkpatrick model}",
      journal = {Journal of Statistical Physics},
         year = 1991,
        month = jan,
       volume = {62},
       number = {1-2},
        pages = {1-19},
          doi = {10.1007/BF01020856},
      }

@article{Jagannathbasco,
author = {Jagannath, Aukosh and Tobasco, Ian},
year = {2017},
month = {04},
pages = {},
title = {Some Properties of the Phase Diagram for Mixed $p$-Spin Glasses},
volume = {167},
journal = {Probability Theory and Related Fields},
doi = {10.1007/s00440-015-0691-z}
}

@article{BarMacrLun20,
author = {Luneau, Clément and Barbier, Jean and Macris, Nicolas},
year = {2020},
month = {09},
pages = {},
title = {Mutual information for low-rank even-order symmetric tensor estimation},
volume = {10},
journal = {Information and Inference: A Journal of the IMA},
doi = {10.1093/imaiai/iaaa022}
}

@article{Dey2020FluctuationRF,
  title={Fluctuation Results for Multi-species Sherrington-Kirkpatrick Model in the Replica Symmetric Regime},
  author={Partha S. Dey and Qiang Wu},
  journal={Journal of Statistical Physics},
  year={2020},
  volume={185}
}

@article{Subag,
author = {Eliran Subag},
title = {{TAP approach for multispecies spherical spin glasses II: The free energy of the pure models}},
volume = {51},
journal = {The Annals of Probability},
number = {3},
publisher = {Institute of Mathematical Statistics},
pages = {1004 -- 1024},
year = {2023},
doi = {10.1214/22-AOP1605},
URL = {https://doi.org/10.1214/22-AOP1605}
}

@article{Bates,
author = {Erik Bates and Youngtak Sohn},
title = {{Free energy in multi-species mixed p-spin spherical models}},
volume = {27},
journal = {Electronic Journal of Probability},
number = {none},
publisher = {Institute of Mathematical Statistics and Bernoulli Society},
pages = {1 -- 75},
year = {2022},
doi = {10.1214/22-EJP780},
URL = {https://doi.org/10.1214/22-EJP780}
}

@article{Baik,
author = {Baik, Jinho and Lee, Ji},
year = {2017},
month = {11},
pages = {},
title = {Free energy of bipartite spherical Sherrington--Kirkpatrick model},
volume = {56},
journal = {Annales de l'Institut Henri Poincaré, Probabilités et Statistiques},
doi = {10.1214/20-AIHP1062}
}

\end{document}